\documentclass[11pt]{article}
\usepackage{graphicx,epsfig}
\usepackage{amsmath, amssymb, amsthm}

\makeatletter
\renewcommand\section{\@startsection {section}{1}{\z@}%
                                   {-3.5ex \@plus -1ex \@minus -.2ex}%nn
                                   {2.3ex \@plus.2ex}%
                                   {\normalfont\large\bfseries}}
\renewcommand\subsection{\@startsection{subsection}{2}{\z@}%
                                     {-3.25ex\@plus -1ex \@minus -.2ex}%
                                     {1.5ex \@plus .2ex}%
                                     {\normalfont\bfseries}}

\makeatother

\newcommand{\bea}{\begin{eqnarray}}
\newcommand{\eea}{\end{eqnarray}}
\newcommand{\be}{\begin{equation}}
\newcommand{\ee}{\end{equation}}

\newcommand{\sgn}{\mathrm{sgn}}

\newcommand{\noi}{\noindent}
\newcommand{\non}{\nonumber}

\newcommand{\Tr}{\mathrm{Tr}}
\newcommand{\im}{\mathrm{Im}}
\newcommand{\re}{\mathrm{Re}}

\newcommand{\st}{$S_\mathrm{t}^1\,\,$}
\newcommand{\sm}{$S_\mathrm{M}^1$}

\newcommand{\CN}{\mathcal{N}}

\newcommand{\CP}{\mathcal{P}}

\newcommand{\half}{\textstyle{\frac{1}{2}}}

\theoremstyle{definition}
\newtheorem{definition}{Definition}
\newtheorem{proposition}{Proposition}

\begin{document}

\begin{titlepage}

\begin{center}

\hfill RUNHETC-2009-04

\vskip 2 cm {\Large \bf Stability and duality in
  $\mathcal{N}=2$ supergravity} \vskip 1.25 cm
{Jan Manschot}\\
{\vskip 0.5cm  NHETC, Rutgers University \\ Piscataway, NJ 08854-8019 USA}

\pagestyle{plain}
\end{center}

\vskip 2 cm

\begin{abstract}
\baselineskip=18pt
\noi 
The BPS-spectrum is known to change when moduli cross a wall of marginal
stability. This paper tests the compatibility of
wall-crossing with $S$-duality and electric-magnetic duality for   
$\mathcal{N}=2$ supergravity. To this end, the BPS-spectrum of
D4-D2-D0 branes is analyzed in the large volume limit of Calabi-Yau
moduli space. Partition functions are presented, which capture the 
stability of BPS-states corresponding to two constituents with primitive
charges and supported on very ample divisors in a compact Calabi-Yau. 
These functions are ``mock modular invariant'' and therefore confirm
$S$-duality. Furthermore, wall-crossing preserves electric-magnetic
duality, but is shown to break the ``spectral flow" symmetry of the
$\mathcal{N}=(4,0)$ CFT, which captures the degrees of freedom of a single
constituent. 
\end{abstract}
 
\end{titlepage}

\baselineskip=19pt

\tableofcontents

%%%%%%%%%%%%%%%%%%%%%%%%%%%%%%%%%%%%%%%%%%%%%%%%%%%%%%%%%%%%%%%%%%%%%%%%%%%%%%
\vspace{0.5cm}
\section{Introduction}
The study of BPS-states in physics has been very
fruitful. Their invariance under (part of the)
supersymmetry transformations of a theory makes them insensitive to
variations of certain parameters. This allows the calculation of some
quantities in a different regime than the regime of
interest. BPS-states have been specifically useful in testing various
dualities, for example $S$-duality in 
$\CN=4$ Yang-Mills theory \cite{Vafa:1994tf} or in string theory 
\cite{Sen:1994fa}. Another major application is the 
understanding of the spectrum of supersymmetric
theories of gravity, leading to the microscopic
account of black hole entropy for various supersymmetric black holes
in string theory \cite{Strominger:1996sh, Maldacena:1997de}. 

This article considers the BPS-spectrum of $\CN=2$ supergravity
theories in 4 dimensions. $\CN=2$ supersymmetry is the least amount of
supersymmetry, which allows massive states to be BPS. It appears in string theory by
compactifying the 10-dimensional space-time on a compact 6-dimensional
Calabi-Yau manifold $X$. A large class of BPS-states are formed by
wrapping D-branes around cycles of $X$, which might correspond to
black hole states if the number of D-branes is sufficiently large. The
Witten index $\Omega$ (degeneracy counted with $(-1)^F$) is 
insensitive to perturbations of the string coupling constant 
$g_\mathrm{s}$, and plays therefore a central role in this
paper. It allows to show for certain cases that the magnitude of the index agrees with black hole entropy:
$\log \Omega \sim S_\mathrm{BH} $. The
study of D-branes on $X$ revealed many connections to objects in
mathematics, like vector bundles, coherent sheaves and derived
categories, which helps to understand their nature, see for a review
Ref. \cite{Aspinwall:2004jr}. The index $\Omega$ corresponds from this  
perspective to the Euler number $\chi(\mathcal{M})$ of their
moduli space $\mathcal{M}$ \cite{Vafa:1994tf}, or an analogous but better defined 
invariant like Donaldson-Thomas invariants \cite{Thomas:1998}.  
 
An intriguing aspect of BPS-states is their behavior as a function of
the moduli of the theory. The moduli parametrize the Calabi-Yau $X$
and appear in supergravity as scalar fields. Under variations of the moduli,
conservation laws allow BPS-states to become stable or unstable at
codimension 1 subspaces (walls) of the moduli space. Such changes in
the spectrum indeed occur, and were first observed in 4 dimensions by Seiberg and
Witten \cite{Seiberg:1994rs}. Denef \cite{Denef:2000nb} has given an 
illuminating picture of stability in supergravity as multi black hole
solutions whose relative distances depend on the value of the moduli
at infinity. At a wall,
these distances might diverge or become positive and finite.
The changes in the degeneracies $\Delta \Omega$ at a wall
show the impact on the spectrum of these
processes. Ref. \cite{Denef:2007vg} derives formulas for $\Delta
\Omega$ for $n$-body semi-primitive decay using arguments from
supergravity. The notion of stability for D-branes is
closely related to the notion of stability in mathematics
\cite{Douglas:2000ah,Douglas:2000gi}. In this context,
Kontsevich and Soibelman \cite{Kontsevich:2008} derive a very general wall-crossing formula for 
(generalized) Donaldson-Thomas invariants. Gaiotto {\it et al.} \cite{Gaiotto:2008cd}
shows that this generic formula applied to the indices
of 4-dimensional $\CN=2$ quantum field theory, is implied by
properties of the field theory. 

Much evidence exists for the presence of an $S$-duality and electric-magnetic duality
group in $\CN=2$ supergravity \cite{Bohm:1999uk,de
  Wit:1996ix}. $S$-duality is an $SL(2,\mathbb{Z})$ group which
exchanges weak and strong coupling; electric-magnetic duality
 is the action of a symplectic group on the vector multiplets. 
These dualities impose strong constraints on the spectrum of the
theory. The wall-crossing formulas are very generic on the other hand,
and the walls form a very intricate web in the moduli space. It is
therefore appropriate to ask: {\it are wall-crossing and duality 
  compatible with each other?} This paper analyses this question, concentrating on
D4-D2-D0 BPS-states or M-theory black holes, in the large volume limit
of Calabi-Yau moduli space. The BPS-objects correspond in this limit
to coherent sheaves on a Calabi-Yau 3-fold supported on an ample
divisor. The analysis considers the walls, the primitive wall-crossing formula
and (part of) the supergravity partition function
$\mathcal{Z}_\mathrm{sugra}(\tau,C,t)$, which enumerates the indices as
a function of D2- and D0-brane charges for fixed D4-brane
charge. $\mathcal{Z}_\mathrm{sugra}(\tau,C,t)$ captures the changes in the
spectrum by wall-crossing. $S$-duality predicts modularity for this function, which is
tested in this paper.

The degrees of freedom of a single D4-D2-D0 black hole are
related via M-theory to a 2-dimensional $\CN=(4,0)$ superconformal field theory
(SCFT) \cite{Maldacena:1997de}.  One of the symmetries of the SCFT spectrum is
the ``spectral flow symmetry'' \cite{deBoer:2006vg,
  Gaiotto:2006wm, Kraus:2006nb}, which are certain transformations of
the charges, which do not change the value of the moduli at
infinity. This imposes additional constraints on the spectrum to the ones
imposed by the supergravity duality groups. A single constituent cannot decay any
further, and conjectures by \cite{deBoer:2008fk, Andriyash:2008it}
indicate that the SCFT description of the spectrum (for given charge)
might only be valid for a specific value of the moduli. Therefore,
interesting dependence of the SCFT spectrum as a function of the
moduli at infinity is not expected. This suggests that a natural
decomposition for the supergravity partition function with fixed magnetic charge $P$ might be  
\be
\label{eq:zcftpzwc}
\mathcal{Z}_\mathrm{sugra}(\tau,C,t)=\mathcal{Z}_\mathrm{CFT}(\tau,C,t)+\mathcal{Z}_\mathrm{wc}(\tau,C,t),
\ee
where  $\mathcal{Z}_\mathrm{CFT}(\tau,C,t)$ is the well-studied SCFT
elliptic genus \cite{deBoer:2006vg, Gaiotto:2006wm, Kraus:2006nb,
  Manschot:2008zb}, and all wall-crossing in the moduli space is captured
by $\mathcal{Z}_\mathrm{wc}(\tau,C,t)$. $\mathcal{Z}_\mathrm{CFT}(\tau,C,t)$
is known to transform as a modular form from arguments of CFT; the modular
properties of $\mathcal{Z}_\mathrm{wc}(\tau,C,t)$ are however unknown. 

This paper considers a small part of $\mathcal{Z}_\mathrm{wc}(\tau,C,t)$, namely
$\sum_{P_1+P_2=P \atop \mathrm{ample,\, primitive}}\mathcal{Z}_{P_1\leftrightarrow P_2}(\tau,C,t)$, which enumerates the
indices of composite BPS-configurations with two constituents, with ample and primitive
magnetic charges $P_1$ and $P_2$. An important building block
of these functions is the newly introduced ``mock Siegel-Narain theta
function''. Mock modular forms do not transform exactly as modular
forms, but can be made so by the addition of a relatively simple
correction term \cite{Zwegers:2000}, which is applied to mock Siegel-Narain
theta functions in the appendix. Using its transformation properties,
one can show that the corrected partition function transforms
precisely as the SCFT elliptic genus, thereby confirming $S$-duality. 

From the analysis follows also that electric-magnetic duality remains present in
the theory, but the ``spectral flow'' symmetry of the SCFT is
generically not present. This is not quite unexpected since this is not a symmetry of
supergravity. Another indication that the spectral flow 
symmetry is not present appears in
Ref. \cite{Andriyash:2008it}, which explains that the jump in the D4-D2-D0 index by
wall-crossing can be larger than the index of a single BPS-object
(this effect is known as the entropy enigma \cite{Denef:2007vg}). 

A special property of $\mathcal{Z}_{P_1\leftrightarrow P_2}(\tau,C,t)$ is that
it does not contribute to the index if the moduli are chosen at the
corresponding attractor point.  However,
$\mathcal{Z}_{P_1\leftrightarrow P_2}(\tau,C,t)$ is generically not
zero, and therefore $\mathcal{Z}_\mathrm{sugra}(\tau,C,t)$ is nowhere
equal to $\mathcal{Z}_\mathrm{CFT}(\tau,C,t)$ generically. Section
\ref{sec:infradius} explains how these observations are in agreement
with conjectures of Refs. \cite{deBoer:2008fk, Andriyash:2008it} about
the uplift of these BPS-configurations to five dimensions.

Although the compatibility with the dualities is expected, it is very
interesting to see how it is realized. The stability condition and primitive
wall-crossing formula combine in an almost miraculous way to the mock
Siegel-Narain theta function, which gives insights in the way
wall-crossing is captured by $\CN=2$ BPS partition functions for compact
Calabi-Yau 3-folds. An intriguing property of the corrected partition
function is that it is continuous as a function of the K\"ahler moduli
$t$, which is reminiscent of earlier discussions \cite{Gaiotto:2008cd,
  Joyce:2006pf}.   

The outline of this paper is as follows. Section \ref{sec:sugra}
reviews briefly the relevant aspects of $\mathcal{N}=2$
supergravity. Section \ref{sec:pfunction} describes the BPS-states of
interest and the expected properties of their partition
function. Section \ref{sec:infradius} is the heart of the paper, it
describes the walls and the partition functions capturing 
wall-crossing. Section \ref{sec:conclusion} finishes with
discussions and suggestions for further research. The appendix defines
two mock Siegel-Narain theta functions and gives some of their
properties.  

\section{BPS-states in $\mathcal{N}=2$ supergravity}
\label{sec:sugra}\setcounter{equation}{0}
If IIA string theory is compactified on a compact Calabi-Yau 3-fold $X$, one
obtains $\mathcal{N}=2$ supergravity as the low energy theory in the
non-compact dimensions. The most essential part of the field content
for this article are the $b_2+1$ vector
multiplets, which each contain a $U(1)$ gauge field $F^A_{\mu\nu}$ and
complex scalar $X^A$, $A=1,\dots ,b_2+1$ (with $b_2$ the second Betti number of $X$). The
gauge fields lead to a vector of conserved charges $\Gamma=(P^0, P^a, Q_a,
Q_0)^\mathrm{T}$, $a=1\dots b_2$, which take value in the $(2b_2+2)$-dimensional lattice $L$. The magnetic charges are denoted by $P^A$ and
electric charges by $Q_A$. The charges arise in IIA string theory as
wrapped D-branes on the even homology of $X$; the components of
$\Gamma$ represent 6-, 4-, 2- and 0-dimensional cycles. A symplectic pairing
is defined on the charge lattice 
\be
\left< \Gamma_1 , \Gamma_2 \right>=-P^0_1 Q_{0,2}+P_1\cdot
Q_2-P_2\cdot Q_1+P^0_2Q_{0,1}. \non
\ee
The symplectic inner product is thus
\be
\mathbf{I}=\left(\begin{array}{cccc} & & & -1 \\
& & \mathbf{1} & \\
& -\mathbf{1} & &  \\
1 & & \end{array}\right), \non
\ee
where $\mathbf{1}$ denotes a $b_2\times b_2$ unit matrix.

The scalars $X^A$ parametrize the K\"ahler moduli space of the
Calabi-Yau $X$: the complexified K\"ahler moduli are given by
$t^a=B^a+iJ^a=X^a/X^0$. Here, $B^a$ and $J^a$ are periods of the $B$-field and the K\"ahler
form respectively.\footnote{The moduli $t^a$ will sometimes be viewed as 2-forms instead of
scalars. Similarly, the charges $\Gamma$ can also be viewed as
homology cycles or their Poincar\'e dual forms.} The $B$-field takes
values in $H^2(X,\mathbb{R})$. The K\"ahler forms are restricted 
to the K\"ahler cone $C_X$, which is defined to be the space of 2-forms
such that $\int_\gamma J>0$, $\int_P J^2>0$ and $\int_X J^3>0$ for any
holomorphic curve $\gamma$ and surface $P\in X$. An
accurate Lagrangian description of supergravity requires that the volume of $X$
is parametrically larger than the Planck length, thus $J^a\to
\infty$. This article is mainly concerned with this parameter
regime. Loop and instanton corrections can here be neglected, such that the
prepotential simplifies to the cubic expression
\be
F(t)=\frac{1}{6}d_{abc}t^at^bt^c, \non
\ee
where $d_{abc}$ is the triple intersection number of 4-cycles in $X$.

The supergravity Lagrangian is invariant under the electric-magnetic
duality group, which acts on the vector multiplets and more specifically on
the electric-magnetic fields and moduli. This duality group is
essentially a gauge redundancy, which appears by working on the
universal covering space of the moduli space instead of the moduli
space itself. The group is $Sp(2b_2+2,\mathbb{Z})$: the group of
$(2b_2+2)\times (2b_2+2)$ matrices $\mathbf{K}$ which 
leave invariant $\mathbf{I}$ \cite{de Wit:1996ix}:\footnote{Note that
we use here a different notation as in e.g. \cite{de Wit:1996ix}, which
is more natural from the point of view of geometry.}
\be
\mathbf{K}^\mathrm{T} \mathbf{I} \mathbf{K}= \mathbf{I}. \non
\ee
The arguments that the group is
$Sp(2b_2+2,\mathbb{Z})$ are valid in the large volume limit. The
correct electric-magnetic duality group, which is valid for any value
of $J$, is a subgroup of this and generated by the monodromies
around singularities in the moduli space. These generators are
generically hard to determine, except for the monodromies in the
limit $J\to \infty$. They are simply the
translations\footnote{Note that the upper or lower indices might label either rows or columns in the matrix.} 
\be
\label{eq:periodicity}
\mathbf{K}(k)=\left(\begin{array}{cccc}1 &  & & \\ k^a & \mathbf{1} & & \\ {1\over 2}
  d_{abc}k^bk^c & d_{abc}k^c & \mathbf{1} & \\
{1\over 6} d_{abc}k^ck^bk^c & {1\over 2} d_{abc}k^bk^c & k^a &
1\end{array}\right), \quad k\in \mathbb{Z}^{b_2}.
\ee

In addition, an $SL(2,\mathbb{Z})$ duality group is present, which
exchanges the weak and strong coupling regime. This group acts on the
hypermultiplets, and is most manifest in the IIB description for large 
K\"ahler parameters \cite{Bohm:1999uk}.  

If time is considered as Euclidean and compactified, another
$SL(2,\mathbb{Z})$ duality group appears. This can be seen from the
M-theory viewpoint, where the total geometry is $\mathbb{R}^3\times
T^2\times X$, and $T^2$ is the product of the time and M-theory
circle \st$\times$\sm. A Kaluza-Klein reduction to $\mathbb{R}^3$,
leads to a 3-dimensional $\mathcal{N}=4$ supergravity theory. Since
the physics in $\mathbb{R}^3$ is independent of large coordinate
reparametrizations of $T^2$, it should exhibit an $SL(2,\mathbb{Z})$
duality group. The duality (complex structure) parameter is given by
\be
\label{eq:slparameter}
\tau=C_{1}+ie^{-\Phi}=C_1+i\beta/g_{\mathrm{s}},
\ee
where $C_{1}\in \mathbb{R}$ is the component of the RR-potential
1-form along the time direction. Note that this $SL(2,\mathbb{Z})$ 
exchanges \st  and \sm, which changes the physical interpretation of
the states on both sides of the duality. D2-branes become for example worldsheet
instantons and vice versa. The full BPS-spectrum should however be invariant under
these $SL(2,\mathbb{Z})$ transformations. These transformations also
transform the $B$- and $C$-fields into each other, which is easily
seen from the M-theory perspective: the $B$- and $C$-field are
reductions of the M-theory 3-form over different $2$-cycles of the
torus. The duality transformations are summarized by
\be
\tau \to \frac{a\tau+b}{c\tau +d}, \quad C\to aC + bB, \quad B\to
cC+dB, \quad J\to |c\tau+d|J.
\ee  
with $\left(\begin{array}{cc} a & b \\ c & d \end{array} \right)\in SL(2,\mathbb{Z})$.
Note that this $SL(2,\mathbb{Z})$ is not the weak-strong duality of
the 4-dimensional supergravity. But it is possible to relate this
``M-theory'' $SL(2,\mathbb{Z})$ to the $S$-duality $SL(2,\mathbb{Z})$ of
IIB, by a T-duality along the time circle \cite{Denef:2007vg}. This
transforms $C_1$ into $C_0$ and (\ref{eq:slparameter}) becomes the
familiar IIB duality parameter. The physical D4-D2-D0 branes of IIA
become D3-D1-D-1 instantons of IIB. 
Therefore, a test of the M-theory $SL(2,\mathbb{Z})$ is equivalent to
testing $S$-duality, and in the rest of the paper the M-theory
$SL(2,\mathbb{Z})$ is referred to as $S$-duality.  

The $\mathcal{N}=2$ supersymmetry algebra contains a central element,
the central charge $Z(\Gamma)\in \mathbb{C}$. The central charge of a
BPS-state is a linear function of its charge $\Gamma$ and a non-linear
function of the K\"ahler or complex structure moduli of $X$. Only the
complexified K\"ahler moduli $t^a$ appear in $Z(\Gamma)$ for the relevant
BPS-states in this article, thus $Z(\Gamma,t)$. 

The mass $M$ of supersymmetric states is determined by the supersymmetry
algebra to be $M=|Z(\Gamma,t)|$. In a theory of gravity, a sufficiently massive
BPS-state correspond to a black hole state in the non-compact
dimensions. The moduli depend generically on the spatial position
$t(\vec{x})$ in a black hole solution. Their value at the horizon is determined in terms of the
charge $\Gamma$ by the attractor mechanism \cite{Ferrara:1995ih}, whereas the value at
infinity is imposed as boundary condition. The mass $M$ is determined
by the moduli at infinity. Following sections deal with the stability
of BPS-states, which is determined by these values at infinity. Also the
$SL(2,\mathbb{Z})$ duality group is acting on the complex structure
parameter $\tau$ of $T^2$ at infinity.    

The expression for the central charge as a function of the moduli is
generically highly non-trivial. However in the limit
$J\to \infty$ it simplifies to \cite{Aspinwall:2004jr}
\be
Z(\Gamma,t)=-\int_X e^{-t}\wedge \Gamma, \non
\ee
where the moduli $t$ and the charge $\Gamma$ are viewed as forms on
$X$. Alternatively, one can write  
\be
Z(\Gamma, t)=\left(1, t^a , \textstyle{\frac{1}{2}}d_{abc}t^bt^c,
\textstyle{\frac{1}{6}}d_{abc}t^at^bt^c   \right) \,\mathbf{I}\,
\Gamma= \Pi^\mathrm{T}\,\mathbf{I}\, 
\Gamma, \non
\ee 
where we defined the vector of the periods $\Pi$. 

A very intriguing aspect of BPS-states is their
stability. The simplest example is the case with two
BPS-objects with primitive charges $\Gamma_1$ and $\Gamma_2$. Their
total mass is larger than or equal to the mass of a single BPS-object
with the same total charge: $|Z(\Gamma_1,t)|+|Z(\Gamma_2,t)|\geq 
|Z(\Gamma_1+\Gamma_2,t)|$. The equality is generically not saturated, but for
special values of the moduli $t=t_\mathrm{ms}$, the central charges can
align $Z(\Gamma_1,t_\mathrm{ms})/Z(\Gamma_2,t_\mathrm{ms})\in 
\mathbb{R}^+$, and the equality holds. These values form
a real codimension 1 subspace of the moduli space, appropriately called
the ``walls of marginal stability''. They decompose the moduli space
into chambers. BPS-states might decay or become stable, whenever the
moduli cross a wall. 

Denef \cite{Denef:2000nb} has shown how wall-crossing phenomena are
manifested in supergravity. The equations of motions allow for
BPS-solutions with multiple black holes. The ones of interest for the
present discussion are solutions with only two black holes. The relative
distance between the two centers is given by 
\be
|x_1-x_2|=\sqrt{G_4} \left.\frac{\left<\Gamma_1,\Gamma_2\right>}{2}\frac{|Z(\Gamma_1,t)+Z(\Gamma_2,t)|}{\im
  (Z(\Gamma_1,t)\bar Z(\Gamma_2,t))}\right|_\infty ,\non
\ee
where $|_\infty$ means that the central charges are evaluated at
asymptotic infinity in the black hole solution; $G_4$ is the
4-dimensional Newton constant.\footnote{$G_4$ is the 4-dimensional Newton constant, and is
  given in terms of IIA and M-theory parameters by $G_4=g_\mathrm{s}^2
  \alpha'\frac{(\alpha')^3}{V_{\mathrm{CY}}}$ and $G_4=\ell^2_\mathrm{P}\frac{\ell_\mathrm{P}}{2\pi
  R}\frac{\ell^6_\mathrm{P}}{V_{\mathrm{CY}}}$, respectively.} In the
limit $G_4 \to 0$, or equivalently $g_\mathrm{s}\to 0$, the distance
between the centers also approaches $0$. This is the regime, where a microscopic
analysis is typically carried out, it is the D-brane regime as opposed
to the black hole regime. 

Since distances must be positive, the
solution can only exist for   
\be
\label{eq:stabcondition}
\left<\Gamma_1,\Gamma_2 \right>\im (Z(\Gamma_1,t)\bar Z(\Gamma_2,t))>0.
\ee
Importantly, $|x_1-x_2|$ depends on the moduli: if $t$ approaches a
wall of marginal stability
\be
\label{eq:wall}
\im (Z(\Gamma_1,t)\bar Z(\Gamma_2,t))=0,
\ee
$|x_1-x_2|\to \infty$ and the 2-center solution decays. An
implication of the mechanism for stability in supergravity is that
single center black holes cannot decay into BPS-configurations with
multiple constituents. If the moduli are chosen at the attractor point at infinity, and
are thus constant throughout the black hole solution, 2-center
solutions cannot exist. Moreover, the moduli flow in a 2-center
solution from a stable chamber at infinity, to unstable chambers at the attractor
points.   

In the following, we will analyze wall-crossing between two chambers
$\mathcal{C}_\mathrm{A}$ and $\mathcal{C}_\mathrm{B}$. To avoid ambiguities, one can choose
$\Gamma_1$ and $\Gamma_2$ such that $\im (Z(\Gamma_1,t)\bar
Z(\Gamma_2,t))<0$ in $\mathcal{C}_\mathrm{B}$, which is equivalent to the convention 
in the mathematical literature, see for example \cite{Yoshioka:1994}. This means that a
stable object with charge $\Gamma$ satisfies 
\be
\frac{\im(Z(\Gamma_2,t))}{\re(Z(\Gamma_2,t))}<\frac{\im(Z(\Gamma,t))}{\re(Z(\Gamma,t))},\non
\ee
with $\left<\Gamma,\Gamma_2 \right>>0$.  

Coherent sheaves are expected to be the proper mathematical
description of D-branes in the limit $J\to \infty$
\cite{Aspinwall:2004jr}. The charge $\Gamma$ of the BPS-state is
determined by the Chern character of corresponding sheaf $\mathcal{E}$
and of the $\hat A$ genus of the Calabi-Yau \cite{Minasian:1997mm}
\be
\Gamma=\mathrm{ch}(i_!\mathcal{E})\,\sqrt{\hat A(TX)},
\ee
where $i:P\hookrightarrow X$ is the inclusion map of the divisor into
the Calabi-Yau. 

Of central interest are the degeneracies of BPS-states with charge
$\Gamma$. Most useful is actually the index  
\be
\Omega(\Gamma;t)=\frac{1}{2}\Tr_{\mathcal{H}(\Gamma;t)}\,(2J_3)^2(-1)^{2J_3},
\ee 
where $J_3$ is a generator of the rotation group
Spin(3). $\Omega(\Gamma;t)$ is a protected quantity against variations
of $g_\mathrm{s}$. The degeneracies are only constant in chambers of the moduli 
space, but jump if a wall is crossed. This is easily understood from
the mechanism for decay in supergravity: the constituents separate,
leading to a factorization of the Hilbert spaces, and consequently a
loss of the number of states. The change in the index is
\cite{Denef:2007vg}:   
\be
\label{eq:deltaomega}
\Delta \Omega(\Gamma;t_\mathrm{s}\to t_\mathrm{u})= \Omega(\Gamma;
t_\mathrm{u})-\Omega(\Gamma; t_\mathrm{s})=-(-1)^{\left<\Gamma_1,\Gamma_2
  \right>-1}\left| \left<\Gamma_1,\Gamma_2 \right> \right|
\Omega(\Gamma_1 ; t_\mathrm{ms})\, \Omega(\Gamma_2;t_\mathrm{ms}).\non
\ee
Of course, in crossing a wall towards stability one gains
states. Therefore the change of the index is in this case 
\be
\Delta \Omega(\Gamma;t_\mathrm{u}\to t_\mathrm{s})=(-1)^{\left<\Gamma_1,\Gamma_2
  \right>-1}\left| \left<\Gamma_1,\Gamma_2 \right> \right|
\Omega(\Gamma_1;t_\mathrm{ms})\, \Omega(\Gamma_2;t_\mathrm{ms}).\non
\ee
Wall-crossing occurs more generally between two chambers
$\mathcal{C}_\mathrm{A}$ and $\mathcal{C}_\mathrm{B}$. If $\Gamma_1$
and $\Gamma_2$ are chosen such that $\im
(Z(\Gamma_1,t_\mathrm{B})\bar Z(\Gamma_2,t_\mathrm{B}))<0$ in 
$\mathcal{C}_\mathrm{B}$, the change of the index between the two chambers is 
\be
\label{eq:changeOmega}
\Delta \Omega(\Gamma;t_\mathrm{A}\to t_\mathrm{B})=(-1)^{\left<\Gamma_1,\Gamma_2
  \right>}\left<\Gamma_1,\Gamma_2 \right> \Omega(\Gamma_1;t_\mathrm{ms})\, \Omega(\Gamma_2;t_\mathrm{ms}).
\ee
This is consistent with jumps of the invariants in mathematics at walls
of marginal stability. We can of course choose the points $t_\mathrm{A}$ and
$t_\mathrm{B}$ more generally and allow them to lie in the same
chamber. Then the change in the index is
\begin{eqnarray}
\label{eq:changeOmega2}
\Delta \Omega(\Gamma;t_\mathrm{A}\to t_\mathrm{B})&=&(-1)^{\left<\Gamma_1,\Gamma_2
  \right>}\left<\Gamma_1,\Gamma_2 \right>
\Omega(\Gamma_1;t)\, \Omega(\Gamma_2;t)\\
&&\times \half 
\left(\,\sgn(\im(Z(\Gamma_1;t_\mathrm{A}) \bar
Z(\Gamma_2;t_\mathrm{A})))-\sgn(\im(Z(\Gamma_1;t_\mathrm{B}) \bar
Z(\Gamma_2;t_\mathrm{B})))\,\right),\non 
\end{eqnarray}
where $\sgn(z)$ is defined as $\sgn(z)=1$ for $z> 0$, $0$ for $z=0$,
and $-1$ for $z<0$.  Note that $\Delta \Omega(\Gamma;t_\mathrm{A}\to t_\mathrm{B})$ satisfies a
cocycle relation
$[\mathrm{AC}]=[\mathrm{AB}]+[\mathrm{BC}]$.

Compatibility of the earlier described dualities with wall-crossing is
non-trivial. Consider here the compatibility of electric-magnetic
duality. As a gauge redundancy,
$Sp(2b_2+2,\mathbb{Z})$ (or the relevant subgroup) leaves invariant the central    
charge: $Z(\Gamma; t)=Z(\mathbf{K}\Gamma;\mathbf{K}t)$ ($\mathbf{K}t$ denotes the transformed
vector of moduli), and the indices: 
\be
\label{eq:degperiodicity}
\Omega(\Gamma; t)=\Omega(\mathbf{K}\Gamma;\mathbf{K}t),
\ee
for every $\Gamma\in L$. Since the walls are determined by
the central charges, and $\left<\Gamma_1,\Gamma_2 \right>=\left<\mathbf{K}\Gamma_1,\mathbf{K}\Gamma_2
\right>$, it is clear that wall-crossing does not obstruct the electric-magnetic duality group. 
Note that generically $\Omega(\Gamma;t) \neq
\Omega(\mathbf{K}\Gamma;t)$, and that no symmetry exists in
supergravity which relates these two indices. Section \ref{sec:pfunction} comes back to
this point.  

The $SL(2,\mathbb{Z})$-duality group also implies non-trivial
constraints for the degeneracies and their wall-crossing. The test of
this duality is however much more involved and the subject of Section \ref{sec:infradius},
after general aspects of D4-D2-D0 BPS-states and their partition functions are explained in the
next section.

\section{D4-D2-D0 BPS-states}
\label{sec:pfunction}\setcounter{equation}{0}
This section specializes the general considerations of the previous
section to the set of states with charge $\Gamma=(0,P,Q,Q_0)$, and
discusses the supergravity partition functions for this class of charges. These
BPS-states correspond to D4-branes wrapping a divisor in $X$, with homology
class $P\in H_4(X,\mathbb{Z})$. This class of BPS-states is
well-described in the literature, see for example
\cite{Maldacena:1997de, Minasian:1999qn,
  deBoer:2006vg,Gaiotto:2006wm}, therefore the review here will
only include the most essential parts for the discussion.

The divisor is also denoted by $P$ and taken to be very ample, which means among
others that it has non-zero positive components in all 4-dimensional
homology classes. 
The intersection form on $P$ leads to a quadratic form
$D_{ab}=d_{abc}P^c$ for magnetic charges $k\in H_4(X,\mathbb{Z})$, the
signature of $D_{ab}$ is $(1,b_2-1)$. The lattice is denoted by 
$\Lambda$. The electric charge $Q$ takes its value in
$\Lambda^*+P/2$ \cite{Freed:1999vc,Minasian:1997mm}. The conjugacy class of $Q-P/2$ in $\Lambda^*/\Lambda$  
is denoted by $\mu$. If necessary, the dependence of $D_{ab}$ on $P$ will
be made explicit, like $P\cdot J^2$, otherwise simply $J^2$ is
used. 

The real and imaginary part of the central charge $Z((P,Q,Q_0),t)$ of these
states are 
\begin{eqnarray}
\re(Z(\Gamma,t))&=& \frac{1}{2}P \cdot (J^2-B^2)+Q\cdot B - Q_0 ,\nonumber \\
\im(Z(\Gamma,t))&=& (Q-BP)\cdot J. \nonumber
\end{eqnarray}
The mass $|Z(\Gamma;t)|$ of BPS-states in the regime $P\cdot J^2\gg
|(Q-\frac{1}{2}B)\cdot B-Q_0|,\,|(Q-BP)\cdot J|$ is:
\be
\label{eq:mass}
|Z(\Gamma,t)|= \frac{1}{2}P\cdot J^2+(Q-\frac{1}{2}B P)\cdot
B-Q_0+\frac{((Q-BP)\cdot J)^2}{P\cdot J^2}+\mathcal{O}(J^{-2}).
\ee
All but the first term are homogeneous of degree 0 in $J$, and thus invariant
under rescalings. The combination $\frac{((Q-B P)\cdot J)^2}{P\cdot
  J^2}$ is positive definite: $(Q-B)^2_+$. $J$ has thus a natural
interpretation as a point of the Grassmannian which parametrizes
 1-dimensional subspaces on which $D_{ab}$ is positive
definite. It therefore determines a decomposition of $\Lambda\otimes 
\mathbb{R}$ into a 1-dimensional positive definite subspace and a $(b_2-1)$-dimensional
negative definite subspace.  

For $P^0=0$, $P\neq 0$, the transformations (\ref{eq:periodicity}) act on
the charges and moduli as
\begin{eqnarray}
&Q_0& \to Q_0 + k\cdot Q +\frac{1}{2}d_{abc}k^ak^bP^c, \non \\
&Q_a& \to Q_a + d_{abc}k^bP^c, \non \\
&t^a& \to t^a+k^a, \non 
\end{eqnarray}
with $k^a \in \Lambda$. 

As mentioned in the introduction, the microscopic explanation for the macroscopic entropy
$S_\mathrm{BH}=\pi |Z|^2$ of a single center D4-D2-D0 black hole was
given by Ref. \cite{Maldacena:1997de} using M-theory. The black hole degrees of freedom are in this
case those of an M5-brane which wraps the divisor in $X$ times the
torus $T^2$. The microscopic
counting relied on a 2-dimensional $\mathcal{N}=(4,0)$ CFT, which can
be obtained as the reduction of the M5-brane worldvolume theory to
$T^2$. The magnetic charge $P$ determines mainly the field content of
the CFT, whereas the electric charges $Q$ and $Q_0$ are charges of
states within the CFT. The BPS-indices of the single center black hole
are the Fourier coefficients of the SCFT elliptic genus 
$\mathcal{Z}_\mathrm{CFT}(\tau,C,B)$ \cite{
  deBoer:2006vg,Gaiotto:2006wm, Kraus:2006nb}.     

To test the compatibility of $S$-duality in supergravity with wall-crossing, one needs
to consider the full supergravity partition function
$\mathcal{Z}(\tau,C,t)$,\footnote{The subscript ``sugra'' used in the
  introduction will be omitted.} which captures the stability of BPS-states as
a function of $t$. Properties of $\mathcal{Z}(\tau,C,t)$ are now briefly reviewed,
tailored for the present discussion. It is defined by 
\begin{eqnarray}
&\mathcal{Z}(\tau,C,t)&=\sum_{Q_0,\, Q}\, \Tr_{\mathcal{H}(P, Q, Q_0; t)}\,\textstyle{\frac{1}{2}}(2J_3)^2 (-1)^{2J_3+P\cdot Q}
\non \\
&&\exp\left(-2\pi \tau_2|Z(\Gamma,t)| + 2\pi i \tau_1 (Q_0-Q\cdot B +B^2/2)+2\pi i C\cdot (Q-B/2)\right),\non
\end{eqnarray}
with $\tau_2=\frac{\beta}{g_{\mathrm{s}}}\in\mathbb{R}^+$,
$\tau_1=C_1 \in \mathbb{R}$, $t=B+iJ\in \Lambda \otimes
\mathbb{C}$ and $B,C\in \Lambda \otimes
\mathbb{R}$. This function sums over Hilbert spaces with fixed magnetic
charge and varying electric charges. This is in agreement with a
microcanonical ensemble for magnetic charge and a canonical ensemble
for electric charges, which is natural in the statistical physics of
BPS black holes \cite{Ooguri:2004zv}. After insertion of (\ref{eq:mass}) one finds
\begin{eqnarray}
\mathcal{Z}(\tau,C,t)&=&\exp(-\pi \tau_2 J^2)\, \sum_{Q_0,\, Q}\,\Tr_{\mathcal{H}(P, Q, Q_0; t)}\,\textstyle{\frac{1}{2}}(2J_3)^2 (-1)^{2J_3+P\cdot Q} \nonumber \\
&&\times e\left(-\bar \tau \hat Q_{\bar 0} +  \tau (Q-B)_+^2/2+ \bar \tau
(Q-B)_-^2/2 + C\cdot (Q-B/2)\right),\nonumber 
\end{eqnarray}
with $\hat Q_{\bar 0}=Q_{\bar 0}+\frac{1}{2}Q^2$, $Q_{\bar
  0}=-Q_0$ and $e(x)=\exp(2\pi i x)$. The modular invariant prefactor $\exp(-\pi \tau_2
J^2)$ is in the following omitted. The partition function has an expansion 
\begin{eqnarray}
\label{eq:expansion}
\mathcal{Z}(\tau,C,t)&=& \sum_{Q_{0}, \, Q}\nonumber
\Omega(P,Q,Q_{0};t)\, (-1)^{P\cdot Q} \\
&&\times e\left(-\bar \tau \hat Q_{\bar 0} +  \tau (Q-B)_+^2/2+ \bar \tau
(Q-B)_-^2/2 + C\cdot (Q-B/2)\right)\nonumber. 
\end{eqnarray}
Note that the partition function depends in various ways on the K\"ahler moduli $t$:
they appear in $\Omega(P,Q,Q_{0};t)$, moreover $B$ shifts the
electric charges and $J$ determines the decomposition of the lattice
into a positive and negative definite subspace of $\Lambda \otimes
\mathbb{R}$. The sum over $Q_0$ and $Q$ is unrestricted and might at some
point invalidate the estimate used for (\ref{eq:mass}), even in the
limit $J\to \infty$. To verify that this does not invalidate the
analysis, we compute the term $\mathcal{O}(J^{-2})$. It is given by
$-2 (Q-B)_+^2\, (\hat Q_{\bar 0}-\half (Q-B)_-^2)/P\cdot J^2$. It
follows from the CFT analysis that $\hat Q_{\bar 0}$ is bounded below
for a single constituent, therefore $\hat Q_{\bar 0}-\half (Q-B)_-^2$
is as well. Moreover, the next section shows that $\hat Q_{\bar
  0}-\half (Q-B)_-^2$ is also bounded below for stable bound states of 2 constituents. If $(Q-B)_+^2\, (\hat Q_{\bar 0}-\half
(Q-B)_-^2)$ is $\mathcal{O}(J^2)$, then $|Z(\Gamma,t)|-\half P\cdot J^2$ is at least
$\mathcal{O}(J)$. Contributions to the partition function of states for which the approximations for Eq. (\ref{eq:mass}) are not
satisfied, are thus highly suppressed compared to the states for which
they are satisfied, which shows that the analysis is not invalidated. Also for any given value of the charges, one can always
increase $J$ to sufficiently large values, such that the
approximations are valid. It is very well possible however,
that not the whole partition function has a nice Fourier
expansion. 

It is well known that $\mathcal{Z}(\tau,C,t)$ contains a pole for
$\tau\to i \infty$ and its $SL(2,\mathbb{Z})$ images. It is less clear at
this point whether poles in $B$ or $C$ can appear 
in $\mathcal{Z}(\tau,C,t)$. Examples of CFT's where such poles
appear, are the characters of massless representations of the $\CN=4$ SCFT algebra
\cite{Eguchi:1987wf}, and the sigma model with the non-compact target space $H_3^+$
\cite{Gawedzki:1991yu}. The Fourier expansion of a partition function
with poles depends on the integration contour. This is how
the partition function of dyons in $\CN=4$ supergravity
\cite{Sen:2007vb} captures wall-crossing phenomena. However, the stability
condition (\ref{eq:wall}) for D4-branes on ample divisors show that no wall-crossing as 
function of $C$ is present. Moreover, the partition functions for
bound states of two constituents, derived in the next section,
are not directly suggestive for ``wall-crossing by poles''. Therefore,
in the following is assumed that no poles in $B$ or $C$ are present in
$\mathcal{Z}(\tau,C,t)$.  

The translations $\mathbf{K}(k)$ of the electric-magnetic duality
group imply a symmetry for the partition function. Using
(\ref{eq:degperiodicity}) and assuming the Fourier expansion, one
verifies easily that   
\be
\mathcal{Z}(\tau,C,t) \longrightarrow (-1)^{P\cdot k}\,e(C\cdot
k/2)\,\mathcal{Z}(\tau,C,t), \non
\ee
under transformations by $\mathbf{K}(k)$. Also using
(\ref{eq:degperiodicity}) one can show a quasi-periodicity in $B$:  
\be
\mathcal{Z}(\tau,C,t+k)=(-1)^{P\cdot k}\,e(C\cdot
k/2)\,\mathcal{Z}(\tau,C,t). \non
\ee
Additionally, $\mathcal{Z}(\tau,C,t)$ satisfies a quasi-periodicity in $C$: 
\be
\label{eq:Cperiodicity}
\mathcal{Z}(\tau,C+k,t)=(-1)^{P\cdot k}\,e(-B\cdot
k/2)\,\mathcal{Z}(\tau,C,t).
\ee
These translations are large gauge transformations of $C$. A theta
function decomposition is not implied by the two  
periodicities since the Fourier coefficients $\Omega(\Gamma;t)$
explicitly depend on $B$, and generically 
$\Omega(\mathbf{K}(k)\Gamma;t)\neq\Omega(\Gamma;t)$.  

A distinguishing property of the partition function
for this class of BPS-states is that charges multiply either $\tau$ or
$\bar \tau$, in contrast to for example D2- or D6-brane partition
functions. Additionally, space-time $S$-duality suggests that the
function transforms as a modular form, such that techniques of the
theory modular forms can be usefully
applied. Refs. \cite{Gaiotto:2006wm,Gaiotto:2007cd} present some
coefficients $\Omega((0,1,Q,Q_0);t)$ for several Calabi-Yau 3-folds
with $b_2=1$. These coefficients determine the whole partition
function, and confirm modularity in a non-trivial way. However,
stability phenomena do not occur in the limit $J\to \infty$ for these
Calabi-Yau's, since $b_2=1$. The next section tests modularity, if wall-crossing
is present. 

The arguments from CFT for modularity are very
robust. Refs. \cite{deBoer:2006vg, Gaiotto:2006wm, Manschot:2008zb}
derive that the action of the generators of $SL(2,\mathbb{Z})$ on 
$\mathcal{Z}_\mathrm{CFT}(\tau,C,t)$ is given by:  
\begin{eqnarray}
\label{eq:Ztransform}
&S&:\quad \mathcal{Z}(-1/\tau,-B,C+i|\tau|J)=\tau^\frac{1}{2}\bar
  \tau^{-\frac{3}{2}}\,\varepsilon(S)\,\mathcal{Z}(\tau,C,t),  \\
&T&:\quad \mathcal{Z}(\tau+1,C+B,t)=\varepsilon(T)\,\mathcal{Z}(\tau,C,t), \non 
\end{eqnarray}
where $\varepsilon(T)=e\left(-c_2(X)\cdot P/24\right)$ and
$\varepsilon(S)=\varepsilon(T)^{-3}$ \cite{Denef:2007vg,
  Manschot:2008zb}. Here the analysis of \cite{deBoer:2006vg,
  Gaiotto:2006wm} is adapted to the supergravity point of view
following \cite{Denef:2007vg}. The next section gives evidence that
the same transformation properties continue to hold for the full supergravity
partition function.\footnote{Evidence exists that $\mathcal{Z}(\tau,C,t)$ does
  only transform as (\ref{eq:Ztransform}) under the full group 
  $SL(2,\mathbb{Z})$ if $P$ is prime. Otherwise it transforms as a modular
form of a congruence subgroup, whose level is determined by the
divisors of $P$. Consequently, the rest of the
article  assumes implicitly that $P$ is prime, although it nowhere
explicitly enters the calculations.} Note that $S$-duality is
consistent with the two periodicities mentioned above. The
periodicities and the $SL(2,\mathbb{Z})$ form together a Jacobi group
$SL(2,\mathbb{Z})\ltimes (\mathbb{Z}^{b_2})^2$.

The partition function for single constituents can be decomposed in a vector-valued modular
form and a theta function by arguments from CFT. The
indices of the CFT are indepent of the moduli at infinity:
$\Omega_\mathrm{CFT}(\Gamma;t)=\Omega_\mathrm{CFT}(\Gamma)=\Omega(\Gamma)$, and obey the ``spectral flow symmetry''
$\Omega(\Gamma)=\Omega(\mathbf{K}(k)\Gamma)$.
To see this, recall that the D2-brane charges appear in the CFT in a $U(1)^{b_{2}}$ current
algebra, which can be factored out of the total  
CFT by the Sugawara construction, which implies that the indices
satisfy $\Omega_\mathrm{CFT}(\Gamma)=\Omega_\mathrm{CFT}(\mathbf{K}(k)\Gamma)$ \cite{deBoer:2006vg,
  Gaiotto:2006wm, Kraus:2006nb}. The name ``spectral flow'' comes originally from the SCFT of superstrings. In the
current context, one could see the flow as a flow of the $B$-field. As
mentioned already after Eq. (\ref{eq:degperiodicity}), no
evidence exists that this is a symmetry of the full spectrum of 4-dimensional
supergravity. In fact, Section \ref{sec:infradius} shows 
that wall-crossing is incompatible with this symmetry at generic points of the moduli space. 

Since the spectral flow symmetry is present in the spectrum of a
single D4-D2-D0 black hole, the theta function decomposition is
reviewed here. We define the functions
\begin{eqnarray}
\label{eq:hP}
h_{P,Q-\frac{1}{2}P}(\tau)=\sum_{Q_{\bar 0}} \Omega(P,Q,Q_0)\,q^{Q_{\bar 0}+\frac{1}{2}Q^2},
\end{eqnarray}
Using that $\Omega(\Gamma)=\Omega(\mathbf{K}(k)\Gamma)$, one can show
that the invariants $\Omega(P,Q,Q_0)$ depend only on $\hat
Q_{\bar 0}$ and the conjugacy class $\mu$ of $Q\in\Lambda^*$, thus
$\Omega(P,Q,Q_0)=\Omega_\mu(\hat Q_{\bar 0})$. Therefore, $h_{P,Q-\frac{1}{2}P}(\tau)=h_{P,Q-\frac{1}{2}P+k}(\tau)$ with $k\in 
\Lambda$. This allows a decomposition of $\mathcal{Z}_\mathrm{CFT}(\tau,C,t)$  into
a vector-valued modular form $h_{P,\mu}(\tau)$ and a Siegel-Narain theta
function $\Theta_\mu(\tau,C,B)$: 
\be
\label{eq:Zdecomp}
\mathcal{Z}_\mathrm{CFT}(\tau,C,t)=\sum_{\mu \in \Lambda^*/\Lambda}
\overline{h_{P,\mu}(\tau)}\, \Theta_\mu(\tau,C,B),
\ee
with
\be
\label{eq:thetafunction}
\Theta_{\mu}(\tau,C,B)=\sum_{Q\in\Lambda+P/2+\mu}(-1)^{P\cdot Q} e\left(\tau (Q-B)_+^2/2+ \bar \tau (Q-B)_-^2/2 + C\cdot (Q-B/2)\right).
\ee
The dependence of $\Theta_{\mu}(\tau,C,B)$ on the K\"ahler moduli $J$
is not made explicit. The transformation properties of $\Theta_{\mu}(\tau,C,B)$ are   
\begin{eqnarray}
\label{eq:thetatransform}
&S&:\quad \Theta_\mu(-1/\tau,
  -B,C)=\frac{1}{\sqrt{|\Lambda^*/\Lambda|}}(-i\tau)^{b_2^+/2}(i\bar\tau)^{b_2^-/2}e(-P^2/4)
  \nonumber \\
&&\qquad\qquad \sum_\nu e(-\mu \cdot
  \nu)\Theta_\nu(\tau,C,B),\nonumber \\
&T&:\quad \Theta_\mu(\tau+1,C+B,B)=e\left((\mu+P/2)^2/2
\right)\Theta_\mu(\tau,B,C). \nonumber
\end{eqnarray}
They satisfy in addition two periodicity relations for
$B$ and $C$ with $k\in \Lambda$:
\begin{eqnarray}
&&\Theta_{\mu}(\tau,C,B+k)=(-1)^{k\cdot P}\,e(C \cdot
k/2)\,\Theta_\mu(\tau,C,B),\nonumber \\
&&\Theta_{\mu}(\tau,C+k,B)=(-1)^{k\cdot P}\,e(-B\cdot k/2)\,\Theta_\mu(\tau,C,B). \nonumber
\end{eqnarray}
All the dependence on $\tau$ and the ``explicit'' dependence on $B$,
$C$ and $J$ of $\mathcal{Z}_\mathrm{CFT}(\tau,C,t)$ is captured by the $\Theta_\mu(\tau,C,B)$. Note that the
$\Theta_\mu(\tau,C,B)$ are annihilated by $\mathcal{D}=\partial_\tau
+\frac{i}{4\pi}\partial^2_{C_+}+\frac{1}{2}B_+\cdot\partial_{C_+}-\frac{1}{4}\pi i
B_+^2$. $\mathcal{Z}(\tau,C,t)$ is also annihilated by $\mathcal{D}$, if
holomorphic anomalies in $h_{P,\mu}(\tau)$ are ignored; these are known to
arise in similar partition functions for 4-dimensional gauge theory \cite{Vafa:1994tf}.

The transformation properties
of $\Theta_\mu(\tau,C,B)$ imply that $h_{P,\mu}(\tau)$ transforms as a
vector-valued modular form: 
\begin{eqnarray}
\label{eq:htransform}
&S&: \quad
  h_{P,\mu}(-1/\tau)=-\frac{1}{\sqrt{|\Lambda^*/\Lambda|}}(-i\tau)^{-b_2/2-1}\varepsilon(S)^*e\left(-P^2/4
  \right) \non \\
&& \qquad \qquad \qquad \times \sum_{\delta\in \Lambda^*/\Lambda} e(-\delta\cdot \mu)
h_{P,\delta}(\tau), \non \\
&T&: \quad h_{P,\mu}(\tau+1)=\varepsilon(T)^*e\left((\mu + P/2)^2/2
\right) h_{P,\mu}(\tau). \non
\end{eqnarray}
From the asymptotic growth of these Fourier coefficients follows the
black hole entropy $S_\mathrm{BH}=\pi\sqrt{\frac{2}{3}(P^3+c_2(X)\cdot
  P)\hat Q_{\bar 0}}\,$ for $\hat Q_{\bar 0}\gg P^3+c_2(X)\cdot P$. 

\section{Wall-crossing in the large volume limit}
\label{sec:infradius}\setcounter{equation}{0}
As explained in Section \ref{sec:pfunction}, the partition function is
expected to exhibit the modular symmetry and electric-magnetic duality
in the large volume limit $J\to \infty$. This section constructs the
contribution $\mathcal{Z}_{P_1\leftrightarrow P_2}(\tau,C,t)$ of bound
states of two primitive constituents with primitive D4-brane charges
$P_1$ and $P_2\neq \vec 0$ to $\mathcal{Z}(\tau,C,t)$, and tests its
modular properties. I take the following Ansatz for the contribution to the index of a bound state of two primitive
constituents at a point $t$ in the moduli space:
\begin{eqnarray}
\label{eq:cont12}
\Omega_{\Gamma_1\leftrightarrow\Gamma_2}(\Gamma;t)&=&\half \left(\sgn(\im(Z(\Gamma_1,t)\bar Z(\Gamma_2,t)))+\sgn(\left< \Gamma_1,\Gamma_2\right>)\right)\\
&&\times (-1)^{\left< \Gamma_1,\Gamma_2\right>-1}\left< \Gamma_1,\Gamma_2\right>\Omega(\Gamma_1)\Omega(\Gamma_2). \non
\end{eqnarray}
The first term of the first line ensures that this Ansatz reproduces
the wall-crossing formula (\ref{eq:changeOmega2}). The non-trivial part of the Ansatz is thus the term
$\sgn(\left< \Gamma_1,\Gamma_2\right>)$. This section explains that
this is also in agreement with other important physical
requirements. Based on the Ansatz, the generating function of the contribution to the
index of the bound states is determined in Eq. (\ref{eq:genfunction}). A study of the generating function leads to the following results:
\begin{itemize}
\item [-] the generating function (\ref{eq:genfunction}) is convergent,
\item [-] the generating function does not exhibit the modular
  properties of $\mathcal{Z}_{\mathrm{CFT}}(\tau,C,t)$
  (the partition function of a single center black hole with magnetic
  charge $P_1+P_2$), but it can be made so by the
  addition of a ``modular completion'' using techniques of mock
  modular forms. The ``completed'' generating function (\ref{eq:PFstability}) is proposed as the
  contribution $\mathcal{Z}_{P_1\leftrightarrow  P_2}(\tau,C,t)$ of
  2-center bound states, which is thus compatible with $S$-duality.   
\item [-] $\mathcal{Z}_{P_1\leftrightarrow
    P_2}(\tau,C,t)$ has the unexpected property that it is continuous
  as function of the moduli, which is reminiscent of earlier work on
  wall-crossing \cite{Gaiotto:2008cd,Joyce:2006pf}. The
  generating function is by construction a discontinuous function of
  the moduli.
\end{itemize}
The combination of the first and second property is essentially a
unique consequence of the Ansatz. The agreement of the Ansatz with the
supergravity picture is discussed later.

We continue now by taking a closer look at the
walls of marginal stability. Specializing Eq. (\ref{eq:wall}), gives for the walls at $J\to \infty$
(without $1/J$ corrections)    
\be
\label{eq:wallinfradius}
P_1\cdot J^2\, (Q_2-BP_2)\cdot J - P_2\cdot J^2\,(Q_1-BP_1)\cdot J = 0.
\ee
Note that this wall is independent of the D0-brane charges
$Q_{0,i}$. And so states decay at this wall, independent of their
D0-charge and of their distribution between the constituents. The
condition for stability for this class of states is
\be
P_1\cdot J^2(Q_2-BP_2)\cdot J-P_2\cdot J^2 (Q_1-BP_1)\cdot J<0, \non
\ee
if $\left< \Gamma_1,\Gamma_2\right>>0$. This stability condition is a natural generalization of slope stability for
sheaves or bundles on surfaces \cite{Donaldson:1990}, since $P\cdot J^2$ replaces
the notion of rank. It can be derived from the stability for sheaves
\cite{Huybrechts:1996}. When $1/J$ corrections are included, one finds
that actually many physical walls merge with each other in the limit
$J\to \infty$ \cite{Diaconescu:2007bf}. We define 
\be
\mathcal{I}(Q_1, Q_2;t)=\frac{P_1\cdot J^2\,(Q_2-BP_2)\cdot
J-P_2\cdot J^2\,(Q_1-BP_1)\cdot J}{\sqrt{P_1\cdot J^2\,P_2\cdot
    J^2\,P\cdot J^2}},
\ee
which is invariant under rescalings of $J$.

It is instructive to look at the symmetries of the wall
(\ref{eq:wallinfradius}). Clearly, it is invariant under the
translations $\mathbf{K}(k)$ (\ref{eq:periodicity}), if it acts both
on the charges and the moduli. However, the wall is not invariant in
general if only the charges are transformed. This is only the case for
very special situations like $P_1||P_2$. The change in the index
is therefore not consistent with the spectral flow
symmetry. Indeed, already in Section \ref{sec:sugra} we argued that   
this symmetry is not natural from the supergravity perspective. The
fact that the symmetry is broken has major implications for supergravity
partition functions, since the decomposition into a vector-valued
modular form and theta functions is not valid. 

We can now see that Eq. (\ref{eq:cont12}) is in agreement with the supergravity picture. As
mentioned before, the picture of stability in supergravity shows that only the single center solution 
exists if the moduli are chosen at the corresponding attractor point
$t(\Gamma)$. Therefore, the index should equal the CFT-index at this point:
$\Omega(\Gamma;t(\Gamma))=\Omega_\mathrm{CFT}(\Gamma)$, which is consistent with the account of black hole entropy
\cite{Maldacena:1997de}. More evidence for this idea comes from the
conjectures in Refs. \cite{deBoer:2008fk, Andriyash:2008it}, which suggest 
 a one to one correspondence between connected components of
the solution space of multi-centered asymptotic AdS$_3\times S^2$
solutions and IIA attractor flow trees starting at
$t(\Gamma)=\lim_{\lambda\to \infty}D^{-1}Q + i \lambda
P$. Note that $\mathcal{Z}(\tau,C,t)$ does not depend
on $\lambda$ in the limit $J\to \infty$. By the AdS$_3$/CFT$_2$
correspondence, this also suggests that
$\Omega(\Gamma,t(\Gamma))=\Omega_\mathrm{CFT}(\Gamma)$. If this is
correct, (\ref{eq:cont12}) should not contribute to
$\Omega(\Gamma;t(\Gamma))$. Indeed, computation of
$\mathcal{I}(Q_1,Q_2;t(\Gamma))$ gives
$\sqrt{\frac{P^3}{P_1P^2\,P_2P^2}}(P_1\cdot Q_2-P_2\cdot Q_1)$, and therefore
$\sgn(\mathcal{I}(Q_1,Q_2;t(\Gamma)))-\sgn(\CP\cdot Q )=0$, such that
there is never a contribution from bound states at the attractor point
using this Ansatz. On the other hand, bound states with two constituents for charges $\tilde
\Gamma\neq \Gamma$ might exist at $t(\Gamma)$, and consequently
$\Omega(\tilde\Gamma,t(\Gamma))\neq\Omega_{\mathrm{CFT}}(\tilde\Gamma)$. Therefore,
these considerations of BPS-configurations with two 
constituents show that generically $\mathcal{Z}(\tau,C,t)$ equals nowhere in the moduli space 
$\mathcal{Z}_{\mathrm{CFT}}(\tau,C,t)$.

A special choice of charges is $P_1=\vec 0$, i.e. $\Gamma_1=(0, 0, Q_1,
Q_{0,1})$. If one does not move the moduli outside the K\"ahler cone,
then walls for this choice can not be crossed. To see this, recall that $Q_1$
represents now the support of a coherent sheaf and must therefore
represent a holomorphically embedded D2-brane. Therefore, $Q_1\cdot
J>0$ for $J\in C_X$. The stability condition for 
$(P,Q,Q_{0,1})\to (0,Q_1,Q_{0,1})+(P,Q_2,Q_{0,2})$ is given by 
\be
\label{eq:wallinfradius2}
P\cdot Q_1\, Q_1\cdot J <0,
\ee
which is independent of the $B$-field. Eq. (\ref{eq:wallinfradius2})
may or may not be satisfied for given charges. However, because
$Q_1\cdot J$ cannot change its sign for $J\in C_X$, no walls of
marginal stability are present in the large volume
limit. It is thus consistent to consider only bound states of
constituents with non-zero D4-brane charge.

To construct the generating function, it is covenient to introduce
some notation. For constituent $i=1,2$ with
charge $\Gamma_i$, the corresponding quadratic form is denoted by
$(Q_i)^2_i$ and the conjugacy class of
$Q_i$ in $\Lambda_i^*/\Lambda_i$ is
$\mu_i$. $\left<\Gamma_1,\Gamma_2\right>$ can be written as an
  innerproduct of 2 vectors in $\Lambda_{1}\oplus
  \Lambda_2 \otimes \mathbb{R}$. Define to this end the unit vector
  $\CP=\frac{(-P_2,P_1)}{\sqrt{PP_1P_2}}\in \Lambda_{1}\oplus \Lambda_2\otimes
  \mathbb{R}$, then $(Q_1,Q_2)\cdot
  \CP=Q\cdot \CP=\left<\Gamma_1,\Gamma_2\right>/\sqrt{PP_1P_2}$. In the appendix,
  also $\mathcal{I}(Q_1, Q_2;t)$ is written as an innerproduct. 

Since the wall is independent of the D0-brane charge, the index
$\Omega(P,Q,Q_0;t)$ jumps irrespective of the D0-brane charge. For the
partition function, we only want to keep track of the magnetic charge
of the two constituents and sum over all the electric
charge. Therefore, the contribution to the index $\Omega(P,Q,Q_0;t)$
from bound states of constituents whose D4-brane charges are $P_1$ and
$P_2$ includes a sum over D0- and D2-brane charge:
\begin{eqnarray}
\Omega_{P_1\leftrightarrow P_2} (P,Q,Q_0;t)&=&\sum_{(Q_1,Q_{0,1})+(Q_2,Q_{0,2})=(Q,Q_{0})}\half \left(\sgn(\mathcal{I}(Q_1, Q_2;t))-\sgn(\left< \Gamma_1,\Gamma_2\right>)\right)\non\\
&&\times (-1)^{\left< \Gamma_1,\Gamma_2\right>}(P_1\cdot Q_2-P_2\cdot Q_1)\,\Omega(P_1,Q_1,Q_{0,1})\Omega(P_2,Q_2,Q_{0,2}). \non
\end{eqnarray}
The generating function of $\Omega_{P_1\leftrightarrow P_2} (P,Q,Q_0;t)$
analogous to (\ref{eq:hP}) is $h_{P_1 \leftrightarrow
  P_2,Q-\frac{1}{2}P}(\tau)=\sum_{Q_{0}}\Omega_{P_1\leftrightarrow P_2}(P,Q,Q_0;t)\,q^{-Q_{0}
  +\frac{1}{2}Q^2}$. This can be expressed in terms of the
vector-valued modular forms of the last section:  
\begin{eqnarray}
\label{eq:changeOmegaD4}
&&h_{P_1 \leftrightarrow P_2,Q-\frac{1}{2}P}(\tau)\,q^{-\frac{1}{2}Q^2} \nonumber \\
&&\quad=\sum_{(Q_1,Q_{0,1})+(Q_2,Q_{0,2})=(Q,Q_{0})\atop Q_0 }(-1)^{P_1\cdot Q_2-P_2\cdot Q_1}\,(P_1\cdot
Q_2-P_2\cdot Q_1)\,\Omega(\Gamma_1) \,\Omega(\Gamma_2) \nonumber \\
&&\qquad \qquad \qquad \times \,\textstyle{\frac{1}{2}}(\,\sgn(\mathcal{I}(Q_1,Q_2;t))-\sgn(\left< \Gamma_1,\Gamma_2\right>)\,)\,q^{Q_{\bar
    0,1}+Q_{\bar 0,2}} \\
&&\quad=\sum_{Q_1+Q_2=Q}\textstyle{\frac{1}{2}}(\,\sgn(\mathcal{I}(Q_1,Q_2;t))-\sgn(\left< \Gamma_1,\Gamma_2\right>)\,)\,(-1)^{P_1\cdot Q_2-P_2\cdot
  Q_1} \nonumber  \\
&& \qquad \qquad \qquad \times \,(P_1\cdot
Q_2-P_2\cdot Q_1)\, h_{P_1, \mu_1}(\tau)\,h_{P_2,\mu_2}(\tau)\,
q^{-\frac{1}{2}(Q_1)^2_1-\frac{1}{2}(Q_2)^2_2}. \nonumber  
\end{eqnarray}
Note that the spectral flow symmetry is used here to write
$h_{P_i,\mu_i}(\tau)$ instead of $h_{P_i,Q_i-P_i/2}(\tau)$. Eq. (\ref{eq:changeOmegaD4})
can be seen as a major generalization of a similar formula for
rank 2 sheaves on a rational surface \cite{Gottsche:1998}.

To obtain the full generating function, we have to multiply $\overline{h_{P_1 \leftrightarrow
  P_2,Q-\frac{1}{2}P}(\tau)}$ by 
\be
(-1)^{P\cdot Q} \, e\left(\tau (Q-B)_+^2/2 + \bar \tau (Q-B)_-^2/2 + C\cdot (Q-B/2)\right),
\ee
and sum over $Q\in \Lambda^*$. The various quadratic forms in the exponent combine to  
\begin{eqnarray}
e\left(\tau (Q-B)_{+}^2/2+ \bar \tau \left( (Q-B)_{1\oplus 2}^2-(Q-B)_{+}^2\right)/2 + C\cdot
(Q-B/2)\right), \non 
\end{eqnarray}
where $Q_{1\oplus 2}^2=(Q_1)^2_1+(Q_2)^2_2$. See the appendix for more
explanation of the notation. The term $(Q-B)_{1\oplus
  2}^2-2(Q-B)_{+}^2$, which multiplies $\pi\tau_2$ in the exponent is
not negative definite, but has signature $(1,2b_2-1)$. An unrestricted sum over
all $(Q_1,Q_2)\in \Lambda_1\oplus \Lambda_2$ is therefore clearly
divergent. However, the presence of
$\sgn(\mathcal{I}(Q_1,Q_2;t))-\,\sgn(\CP\cdot Q)$ ensures that the
function is convergent, which follows from Proposition \ref{prop:0} in
the appendix. Thus the stability
condition implies that the quadratic form is negative definite, if
evaluated for stable bound states. Performing the sum over $Q$, one
obtains the generating series: 
\be
\label{eq:genfunction}
\sum_{\mu_{1\oplus
    2}\in \Lambda_{1\oplus 2}^*/\Lambda_{1\oplus
    2}}\overline{h_{P_1,\mu_1}(\tau)}\,\overline{h_{P_2,\mu_2}(\tau)}\,\Psi_{\mu_{1\oplus 
    2}}(\tau,C,B),
\ee
 where $\Lambda_{1\oplus 2}=\Lambda_1\oplus \Lambda_2$, $\mu_{1\oplus
  2}=(\mu_1,\mu_2)\in \Lambda_{1\oplus 2}^*/\Lambda_{1\oplus 2}$  and
\begin{eqnarray}
\label{eq:Zterm}
\Psi_{\mu_{1\oplus 2}}(\tau,C,B)&=&\sum_{{Q_1\in \Lambda_1+\mu_1+P_1/2 \atop
      Q_2\in \Lambda_2+\mu_2+P_2/2}}(P_1\cdot
Q_2-P_2\cdot Q_1)\, (-1)^{P_1\cdot Q_1+P_2\cdot Q_2} \non  \\ 
&&\textstyle{\frac{1}{2}} \left(\,
\sgn(\mathcal{I}(Q_1,Q_2;t))-\sgn(\CP\cdot Q )\,\right)   \\
&&\times e\left(\tau (Q-B)_{+}^2/2+ \bar \tau \left( (Q-B)_{1\oplus
  2}^2-(Q-B)_{+}^2\right)/2+C\cdot (Q-B/2)\right), \non 
\end{eqnarray}
with $\CP=\frac{(-P_2,P_1)}{\sqrt{PP_1P_2}}\in \Lambda_{1\oplus
  2}\otimes \mathbb{R}$. 

 The test of $S$-duality is now reduced to testing modularity for
 (\ref{eq:Zterm}). Since $\Psi_{\mu_{1\oplus 2}}(\tau,C,B)$ is not a
 sum over the total lattice $\Lambda_{1\oplus 2}$, it does not have
 the nice modular properties of the familiar theta functions. However,
 Ref. \cite{Zwegers:2000} explains that a real-analytic term can be
 added to a sum over a positive definite cone in an indefinite lattice
 with signature $(n-1,1)$, such that the resulting function transforms
 as a familiar theta function. Appendix \ref{ap:indeftheta}
 applies this technique to $\Psi_{\mu_{1\oplus     2}}(\tau,C,B)$, and
 explains in detail how it can be completed to a function  $\Psi_{\mu_{1\oplus
     2}}^*(\tau,C,B)$, which transforms as a Siegel-Narain theta
 function.\footnote{Note that the Fourier    expansion
   (\ref{eq:expansion}) is thus not modular.} The essential idea of
 this procedure is to make the replacement
\be
\sgn(z)\quad \longrightarrow \quad 2\int_0^{\sqrt{2\tau_2}z}e^{-\pi u^2}du,
\ee
which interpolates monotonically and continuously between $-1$ at
$z=-\infty$ and $1$ at $z=+\infty$. It approaches $\sgn(z)$ in the
limit $\tau_2\to\infty$. To complete $\Psi_{\mu_{1\oplus
    2}}(\tau,C,B)$ to a modular function, one also needs to replace
$z\,\sgn(z)$ by an appriopriate continuous function as explained in
the appendix. Indefinite theta functions are prominent in the work on mock modular forms
 \cite{Zwegers:2000}; $\Psi_{\mu_{1\oplus 2}}(\tau,C,B)$ is therefore appropriately called a
``mock Siegel-Narain theta function''. 

By replacing $\Psi_{\mu_{1\oplus 2}}(\tau,C,B)$ with
$\Psi_{\mu_{1\oplus 2}}^*(\tau,C,B)$ in Eq. (\ref{eq:genfunction}), we
obtain our final proposal of the contribution of 2-center bound states
$\mathcal{Z}_{P_1\leftrightarrow P_2}(\tau,C,t)$ to $\mathcal{Z}(\tau,C,t)$:
\be
\label{eq:PFstability}
\mathcal{Z}_{P_1\leftrightarrow P_2}(\tau,C,t)=\sum_{\mu_{1\oplus 
    2}\in \Lambda_{1\oplus 2}^*/\Lambda_{1\oplus
    2}}\overline{h_{P_1,\mu_1}(\tau)}\,\overline{h_{P_2,\mu_2}(\tau)}\,\Psi_{\mu_{1\oplus 
    2}}^*(\tau,C,B).
\ee
From the transformation properties of the three functions follows that
$\mathcal{Z}_{P_1\leftrightarrow P_2}(\tau,C,t)$ transforms precisely
as the CFT partition function $\mathcal{Z}_{\mathrm{CFT}}(\tau,C,t)$ of the single
constituent with D4-brane charge $P_1+P_2$ (\ref{eq:Ztransform})! To see that the
weight agrees, note that the weight of $\Psi_{\mu_{1\oplus 
    2}}^*(\tau,C,B)$ is $\half(1,2b_2+1)=\half(1,2b_2-1)+(0,1)$, where $\half(1,2b_2-1)$
is due to the lattice sum and $(0,1)$ is due to the insertion of
$P_1\cdot Q_2-P_2\cdot Q_1$. Combining this with $2\cdot (0,-\half
b_2-1)$ of the vector-valued modular forms $\overline{
  h_{P_i,\mu_i}(\tau)}$, one precisely finds the weight
$(\half,-\frac{3}{2})$ for $\mathcal{Z}_{P_1\leftrightarrow
  P_2}(\tau,C,t)$. A crucial detail is the
grading by $(-1)^{P\cdot Q}$: $(-1)^{(P_1+P_2)\cdot (Q_1+Q_2)+(P_1\cdot
  Q_2-P_2\cdot Q_1)}=(-1)^{P_1\cdot Q_1+P_2\cdot Q_2}$, such that $\Psi_{\mu_{1\oplus 
    2}}^*(\tau,C,B)$ does transform conjugately to $\overline{
  h_{P_1,\mu_1}(\tau)}\,\overline{
  h_{P_2,\mu_2}(\tau)}$. Moreover, as was already mentioned
above, coexistence of convergence and modularity is essentially a
unique consequence of the Ansatz. In particular, the fact that
$\CP$ is independent of the moduli and satisfies $\CP\cdot (J,J)=\CP\cdot (B,B)=0$ is
essential. We thus observe that all factors in (\ref{eq:cont12}) combine in a neat way
such that $\mathcal{Z}_{P_1\leftrightarrow 
  P_2}(\tau,C,t)$ has the same modular properties as $\mathcal{Z}_{P_1+  P_2}(\tau,C,t)$. 

One could of course object to correcting the partition function by hand
and argue that an anomaly appeared for $S$-duality. However, the
correcting factor could also arise automatically in a more 
physical derivation, for example by perturbative contributions. It is also not so surprising that corrections to
the Fourier expansion (\ref{eq:expansion}) are necessary, since it was
derived by assuming that the charges are finite and $J\to \infty$, which is
clearly not the case everywhere in the Hilbert space. Note that a physical
derivation might lead to a slightly different modular completion of
the generating function, since one could always add a real-analytic
function with the same transformation properties. This would however not change the crucial properties we have
established.  

Besides $S$-duality, there is another very appealing aspect in favor of the correction
term. Eq. (\ref{eq:Zterm}) is not continuous as a function of
the moduli $B$ and $J$ because of the terms
$\sgn(\mathcal{I}(Q_1,Q_2;t))$. As discussed above, the 
correction term is essentially a replacement of the discontinuous
functions $\sgn(z)$ and $z\,\sgn(z)$ by real analytic functions (which approach the original expression in the limit
$|z|\to \infty$). The modular invariant partition function is
therefore continuous in $B$ and $J$. This might not be such a
coincidence as it seems at first sight. Ref. \cite{Joyce:2006pf}
proposed a continuous and holomorphic generating function for Donaldson-Thomas
invariants (or an extension thereof), which captures wall-crossing. Moreover, Ref. \cite{Gaiotto:2008cd}
describes that continuity of the metric $g$  of the target manifold of
a 3-dimensional sigma model, essentially implies the
Kontsevich-Soibelman wall-crossing formula. Continuity of
$\mathcal{Z}(\tau,C,t)$ is very intriguing from this perspective, and
it would be interesting to investigate whether it plays here an as
fundamental role as in these references.

The contribution of all 2-constituent BPS-states with primitive, ample charges 
is easily included in $\mathcal{Z}(\tau,C,t)$ by the sum $\sum_{P_1+P_2=P
  \atop \mathrm{ample,\, primitive}}\mathcal{Z}_{P_1\leftrightarrow
  P_2}(\tau,C,t)$. The above analyses gives some evidence that
modularity is also preserved if one of the charges is
not ample.

\section{Conclusion and discussion}
\label{sec:conclusion}
\setcounter{equation}{0}
The consistency of wall-crossing with $S$-duality and
electric-magnetic duality is tested by analyzing the BPS-spectrum of
D4-D2-D0 branes on a compact Calabi-Yau 3-fold $X$. The stability of
composite BPS-states with two primitive constituents is considered, in
the large volume limit of the K\"ahler moduli space. The consistency
of electric-magnetic duality with wall-crossing follows rather
straightforwardly from the structure of the walls and the primitive
wall-crossing formula. From the equations for the walls in the 
moduli space can also be seen that wall-crossing is not compatible
with the spectral flow symmetry, which appears in the microscopic
description of a single D4-D2-D0 object by a CFT \cite{Maldacena:1997de}.
$S$-duality is tested by the construction of a partition function
(\ref{eq:PFstability}) for two constituents, which captures the
changes of the spectrum if walls of marginal stability are
crossed. The essential building block is a ``mock Siegel-Narain theta
function'', which might be of independent mathematical interest. The
stability condition and the BPS-degeneracies combine in a very intricate
way in order to preserve modularity, which is a confirmation of
$S$-duality.  

The results of this paper are applicable to various problems,
for example those related to entropy enigmas \cite{Denef:2007vg}. With
these are meant BPS-configurations with multiple constituents, whose
number of degeneracies is larger than the number of degeneracies of a
single constituent with the same charge. Originally, the common
thought was that wall-crossing would only have a subleading effect on the
degeneracies. Ref. \cite{Andriyash:2008it} has shown that enigmatic
changes in the spectrum can also happen from D4-D2-D0 configurations
with 2 constituents, which are considered in this paper. The present
work shows that these enigmatic phenomena, can be captured by modular
invariant partition functions. This might proof useful in future
studies on the entropy enigma. For example Eq. (\ref{eq:PFstability})
shows that the leading entropy of two constituents (if their bound state
exists) is $\pi
\sqrt{\frac{2}{3}(P_1^3+P_2^3+c_2\cdot P)\left(Q_{\bar 0}+\frac{1}{2}(Q_1)_1^2+\frac{1}{2}(Q_2)^2_2\right)}$
extremized with respect to $Q_1$ and $Q_2$, under the constraint
$Q_1+Q_2=Q$. This should be compared with the single constituent
entropy $\pi\sqrt{\frac{2}{3}(P^3+c_2\cdot P)(Q_{\bar
    0}+\frac{1}{2}Q^2)}$. Based on these equations, one can show the
existence of enigmatic configurations, even in the regime
$\sqrt{\frac{\hat Q_{\bar 0}}{P^3}}\gg 1$, or large topological
string coupling. This shows that $\mathcal{Z}_{\mathrm{wc}}(\tau,C,t)$ is not necessarily a small
correction to $\mathcal{Z}_{\mathrm{CFT}}(\tau,C,t)$ in (\ref{eq:zcftpzwc}). A detailed analysis of the
conditions for the first entropy to be larger than the second would
be very instructive.   This raises the question of the relation of the
discussed partition functions in this paper and the
OSV-conjecture, which relates the black hole partition function and
the one of topological strings \cite{Ooguri:2004zv}. 

The D4-D2-D0 BPS-degeneracies are also related to mathematically defined
invariants. In the large volume limit, the D4-D2-D0 index
correspond to the Euler number (or a variant thereof) of the moduli
space of coherent sheaves with support on the divisor of the
Calabi-Yau. An explicit calculation of these Euler numbers is
currently not feasible, but would be magnificent. It would for example
provide a more rigorous test of modularity of the partition
functions. A more tractable possibility for future work is to replace
the index $\Omega(\Gamma;t)$ by a more refined quantity
\cite{Dimofte:2009bv} by including the spin dependence
$\Omega(\Gamma;t,y)=\Tr_{\mathcal{H}(\Gamma;t)}\, (-y)^{2J_3}$. This is not a protected quantity, but is nevertheless of interest. The corresponding
partition function might still exhibit modular properties, and
wall-crossing formulas do exist in the literature for
$\Omega(\Gamma;t,y)$  in the context of surfaces
\cite{Gottsche:1998,Yoshioka:1994} and also physics
\cite{Diaconescu:2007bf}. A generalization of Section \ref{sec:infradius}
 to include these refined invariants should therefore be possible.
 Another suggestion is to move away from the limit $J\to
\infty$ by including finite size corrections. This would also leave
the description of the BPS-states as coherent sheaves, and the 
relations with dualities probably become probably more intricate. 
   
A limitation of this work is that it considers only primitive
wall-crossing. One might continue in a similar fashion as Section
\ref{sec:infradius} to construct partition functions for
BPS-configurations with more constituents, and test the compatibility
of the semi-primitive wall-crossing formula \cite{Denef:2007vg} and
$S$-duality in this way. Much more appealing would be a closed expression
for the partition function, which does not sum over all possible
decays. Such an expression might ultimately allow for a test  
of the generic Kontsevich-Soibelman wall-crossing formula with respect to
$S$-duality. Or even explain the KS-formula in $\CN=2$
supergravity from physical considerations, as was done for $\CN=2$
field theory \cite{Gaiotto:2008cd}. Although this paper took in some
sense an opposite approach, some lessons might still be learned. 

The requirement of the dualities implies non-trivial constraints for the
indices and wall-crossing formulas. These do not seem constraining enough to deduce the
KS-formula. For example, the appearance of mock modular forms instead
of normal modular forms was a priori unknown. This can of course
be seen as an anomaly for $S$-duality. On the other hand,
it is really pretty close to modularity, and the functions can be made
modular by a simple modification as explained in the appendix. These
modifications might appear in a more physical derivation of the
partition function in order to preserve $S$-duality. The correction
terms might be determined by a differential equation, similar to the
holomorphic anomaly equation of topological strings
\cite{Bershadsky:1993cx}. Proposition \ref{prop:5} gives the action of
$\mathcal{D}$, defined in Section \ref{sec:pfunction}, on
$\Psi_{\mu_{1\oplus 2}}^*(\tau,C,B)$. This shows that
$\mathcal{D}\mathcal{Z}_{P_1\leftrightarrow P_2}(\tau,C,t)$ includes a
term $\mathcal{Z}_{\mathrm{CFT},P_1}(\tau,C,B)\mathcal{Z}_{\mathrm{CFT},P_2}(\tau,C,B)$, which
is suggestive and reminiscent of earlier work on holomorphic anomaly equations, see for example Ref. \cite{Minahan:1998vr}. Another consequence of the correction terms
is that they make the function continuous as a function of the moduli, 
although it captures the changes of the spectrum under variations of the
moduli. This is quite intriguing, since ``continuity''  
was essential in the field theory derivation of the KS-formula in
Ref. \cite{Gaiotto:2008cd}, more precisely the continuity of the
metric of the target space of a 3-dimensional sigma
model. The appearance of a continuous partition
function in this paper suggests that continuity might be fundamental
here too. More investigation is clearly necessary to find out to what
extent continuity and the dualities can imply the generic
wall-crossing formula \cite{Kontsevich:2008} for
BPS-invariants.  Ref. \cite{Joyce:2006pf} suggested earlier a
continuous, holomorphic generating function for Donaldson-Thomas
invariants, and its discussion resembles in some respects
Ref. \cite{Gaiotto:2008cd}. However, $\mathcal{Z}_{P_1\leftrightarrow
  P_2}(\tau,C,t)$ does not seem to be holomorphic in $t$.  

Note that the way $\mathcal{Z}_{P_1\leftrightarrow P_2}(\tau,C,t)$
captures stability is quite different from how the partition function
of $\frac{1}{4}$-BPS states (or dyons) of $\CN=4$ supergravity captures
stability. That function captures wall-crossing in a very appealing
way by poles \cite{Sen:2007vb} and a proper choice of the integration
contour \cite{Cheng:2007ch} to obtain Fourier coefficients. In this
way, mock modular forms arise via meromorphic Jacobi forms \cite{Dabholkar:2009}. 

Section \ref{sec:infradius} shows that the supergravity
partition function is nowhere in moduli space equal to the CFT
partition function (except for special cases like a Calabi-Yau with
$b_2=1$). A natural question is: is the supergravity partition
function related to the partition function of a lower dimensional 
theory, just as the spectrum of a single constituent is captured by
the $\CN=(4,0)$ SCFT? Ref. \cite{deBoer:2008fk} (see also \cite{deBoer:2008ss}) proposes that such a
theory might be classically a 2-dimensional sigma model into the moduli space of
supersymmetric divisors in the Calabi-Yau, whose ``beta function does
not vanish for $Y$ \footnote{$Y$ is the
  vector of normalized 5-dimensional K\"ahler moduli, which is proportional to $J$.}
different from the attractor point and the $Y$ undergo renormalization
group flow till they reach the attractor point, an IR fixed
point. Along the flow, the constituents of M5-M5 bound states decouple from
each other; each of them has its own IR fixed point corresponding to
an AdS$_3\times S^2$.'' The structure of the partition function (\ref{eq:PFstability})
shows the decoupled constituents. It is also in agreement with the
suggestion that the theory is not a CFT, since the spectral flow
symmetry is not present. On the other hand, $\mathcal{Z}_{\mathrm{sugra}}(\tau,C,t)$
does not equal $\mathcal{Z}_\mathrm{CFT}(\tau,C,t)$ at attractor
points, which indicates that the microscopic theory (if it exists) is
not a CFT, not even at these points. A better understanding of these
issues is clearly desired. Another alternative for a microscopic
theory is quiver quantum mechanics \cite{Denef:2002ru}, which arises
in the limit $g_\mathrm{s}\to 0$, and is known to capture bound states
in 4 dimensions. A connection between this theory, the D4-D2-D0 bound
states and their partition functions might lead to interesting insights.

An intriguing implication of the proposed function is wall-crossing
as a function of the $C$-field for the BPS-states one obtains after
$S$-duality. A D4-D2-D0 BPS-state becomes a D3-D1-D-1 instantonic 
BPS-state after performing a T-duality along the time circle. This
does not yet change anything fundamental, stability of this
configuration is still captured by $B$ and $J$. However, $S$-duality
transforms such a configuration to one with instanton
D3-branes and fundamental string instantons. Moreover, $B$ and $C$ are
interchanged, which implies that the degeneracies of these BPS-states
jump as a function of $C$ and $J$. This is quite interesting since the
$C$-field is generically not considered as a stability parameter, and 
gives also evidence that $B$ and $C$ should be considered on a more equal
footing. The K-theoretic description of the $C$-fields is however very
different in nature than the description of the $B$-field.

\bigskip
\begin{center}{\bf Acknowledgements}\end{center}
I would like to thank Dieter van den Bleeken, Wu-yen Chuang, Atish
Dabholkar, Emanuel Diaconescu, Davide Gaiotto, Lothar G\"ottsche and Gregory Moore for
fruitful discussions. I owe special thanks to Gregory Moore for
his comments on the manuscript. This work is supported by the DOE under grant
DE-FG02-96ER40949.

\appendix
\section{Two mock Siegel-Narain theta functions}
\label{ap:indeftheta}\setcounter{equation}{0}
This appendix computes the transformation properties of the
Siegel-Narain mock theta function which appears in Section
\ref{sec:infradius}. The proofs are similar to those given in
\cite{Zwegers:2000}. The dependence on the
Grassmannian, which parametrizes 1-dimensional positive definite
subspaces in the lattice $\Lambda$, however complicates the discussion.
First, properties of a simpler 
mock Siegel-Narain theta function are analyzed before those of
$\Psi^{*}_{\mu_{1\oplus 2}}(\tau,C,B)$.  

Let $\Lambda$, $\Lambda_1$ and $\Lambda_2$ be three lattices with
signature $(1,b_2-1)$. The quadratic forms of the lattices are
determined by a cubic form $d_{abc}$: respectively $d_{abc}P^c$, $d_{abc}P^c_1$
and $d_{abc}P^c_2$. The vectors $P_{(i)}$ are characteristic vectors
of the lattices and positive: $P^3_{(i)}>0$. They are related by
$P=P_1+P_2$. The projection of a vector $x\in \Lambda  
\otimes \mathbb{R}$  on the positive definite subspace is determined
by the vector $J\in \Lambda \otimes\mathbb{R}$: $x_+=(x\cdot J/P\cdot J^2)J$,
$x_-=x-x_+$, and $x^2=x_+^2+x_-^2$. The positive definite combination
$x_+^2-x_-^2$ is called the majorant associated to $J$. It is sufficient for this appendix that $J$ lies in the space
\be
C_\Lambda:=\left\{ J\in \Lambda \otimes \mathbb{R}: P_{(i)}\cdot
J^2,\,P_{(i)}^2\cdot J >0,\, i=1,2 \right\}. \non
\ee
$J$ is thus positive in all three lattices.

The direct sum $\Lambda_1\oplus \Lambda_2$ is denoted by $\Lambda_{1\oplus
  2}$ with quadratic form $Q^2_{1\oplus
  2}=(Q_1)^2_1+(Q_2)_2^2$ for $Q=(Q_1,Q_2)\in
\Lambda^*_{1\oplus 2}$. Vectors in $\Lambda_{1\oplus 2}$ are sometimes given the subscript
$1\oplus 2$, and in $\Lambda_i$ the subscript $i$. For example, $P_{1\oplus 2}=P_1+P_2\in
\Lambda_{1\oplus 2}$. Similarly, $\mu_{1\oplus 2}=\mu_1+\mu_2\in
\Lambda^*_{1\oplus 2}/\Lambda_{1\oplus 2}$, and $\mu=\mu_1+\mu_2\in
\Lambda^*/\Lambda$ with $\mu_i\in
\Lambda^*_i/\Lambda_i$. With a slight abuse of notation
$Q_+^2$ denotes $((Q_1+Q_2)\cdot J)^2/P\cdot J^2$.  

Define $\mathcal{I}(Q_1,Q_2;t)$ as in the main text by
\be
\mathcal{I}(Q_1,Q_2;t)=\frac{P_1\cdot J^2 (Q_2-P_2B)\cdot
  J-P_2\cdot J^2 (Q_1-P_1B)\cdot J}{\sqrt{P_1\cdot J^2\,P_2\cdot J^2\,P\cdot J^2}}.
\ee
Define additionally the vector
\be
\mathcal{P}=\frac{(-P_2,P_1)}{\sqrt{PP_1P_2}}\in \Lambda_{1\oplus
  2}\otimes \mathbb{R},
\ee
which satisfies $\mathcal{P}^2=1$.

\begin{definition}
\label{def:1}
Let $t=B+iJ$, with $B\in \Lambda \otimes \mathbb{R}$, and $J\in
C_\Lambda$. Then $\Phi_{\mu_{1\oplus 2}}^{*}(\tau, C, B)$ is  
defined by:
\begin{eqnarray}
\label{eq:def1}
\Phi_{\mu_{1\oplus 2}}^{*}(\tau, C, B)&=&\textstyle{\frac{1}{2}}\sum_{Q\in \Lambda_{1
  \oplus 2}+\mu_{1\oplus 2}+P_{1\oplus 2}/2} \,(-1)^{P_{1}\cdot Q_1+P_{2}\cdot Q_2}\non \\ \,
&&\left(\, E\left(\mathcal{I}(Q_1,Q_2;t)
\sqrt{2\tau_2}\right)-E\left(\mathcal{P}\cdot Q
\sqrt{2\tau_2}\right)\,\right) \\
&&\,\,\times e\left(\tau (Q-B)_+^2/2+\bar \tau ((Q-B)^2_{1\oplus
  2}-(Q-B)_+^2)/2+(Q-B/2)\cdot C \right), \nonumber
\end{eqnarray}
with 
\be
\label{eq:Ez}
E(z)=2 \int^z_0 e^{-\pi u^2}du=\sgn(z)\left(1-\beta(z^2)\right), \non
\ee
where
\be
\beta(x)=\int_x^\infty u^{-\frac{1}{2}}\,e^{-\pi u}\,du,\qquad x\in
\mathbb{R}_{\geq 0}\,.\non
\ee
The moduli in the exponent of (\ref{eq:def1}) are determined by $t$. 
The ``\,*\,'' of $\Phi^*_{\mu_{1\oplus 2}}(\tau, C, B)$ distinguishes
this function from $\Phi_{\mu_{1\oplus 2}}(\tau, C, B)$,  which would
be defined by replacing $E(z)$ by $\sgn(z)$ in the definition.

\end{definition}

\begin{proposition}
\label{prop:0}
$\Phi^*_{\mu_{1\oplus 2}}(\tau, C, B)$ is convergent for $J\in C_\Lambda$
  and $B,C\in \Lambda \otimes \mathbb{R}$.
\end{proposition}

\begin{proof}
First consider the case $B=C=0$. The term which multiplies
$\tau_2$ in the exponent, and thus determines the absolute value of the exponential is 
\be
Q_J^2:=Q_{1\oplus 2}^2-2\frac{((Q_1+Q_2)\cdot J)^2}{P\cdot J^2}=Q_{1\oplus 2}^2-2Q_+^2.
\ee
The signature of this quadratic form is
$(1,2b_2-1)$ which is problematic for convergence.  

To show convergence, note that $0\leq\beta(x)\leq e^{-\pi x}$ for all
$\mathbb{R}_{\geq 0}$ and that therefore the terms involving
$\beta(x)$ in (\ref{eq:def1}) are convergent. Consider next the terms
with
$\sgn(\CP\cdot Q)-\sgn(\mathcal{I}(Q_1,Q_2;iJ))$. There
are essentially two possibilities:
$\sgn(\CP\cdot Q)\,\sgn(\mathcal{I}(Q_1,Q_2;iJ))<0$ or
$>0$. Define the vector 
\be
s(J)=\frac{(-P_2\cdot J^2\, J, P_1\cdot J^2\,J)}{\sqrt{P_1\cdot J^2\,P_2\cdot J^2\,P\cdot J^2}}\in \Lambda_{1\oplus
  2}\otimes \mathbb{R}, \non
\ee
such that $Q\cdot s(J)=\mathcal{I}(Q_1,Q_2;iJ)$ and $s(J)^2=1$.

One can show that $\CP\cdot s(J)=\sqrt{\frac{P\cdot J^2\,
    (P_1P_2J)^2}{PP_1P_2\,P_1\cdot J^2\,P_2\cdot J^2}}>0$ and $\CP_+=s(J)_+=0$. The space
$\mathrm{span}(\CP,s(J))$ has signature $(1,1)$ in $\Lambda_{1\oplus
2}$ with inner product $Q^2_J$. Therefore
\be
\left|\begin{array}{cc}  1 & \CP\cdot s(J) \\
 \CP\cdot s(J) &  1\end{array}\right|=1-(\CP\cdot
s(J))^2< 0 \non.
\ee
Take now a vector $Q\in \Lambda_{1\oplus 2}$,
which is linearly independent of $\CP$ and $s(J)$, then
$\mathrm{span}(Q,\CP,s(J))$ is a space with signature $(1,2)$. Therefore,
\be
\left|\begin{array}{ccc} Q_J^2 & Q\cdot \CP  & Q\cdot s(J)  \\
Q\cdot \CP  & 1 & \CP\cdot s(J) \\
Q\cdot s(J)  & \CP\cdot s(J) &
1 \end{array}\right|>0. \non
\ee
From this follows directly
\be
Q^2_J+\frac{2\,\CP\cdot
  s(J)}{1-(\CP\cdot
  s(J))^2}Q\cdot \CP\,Q\cdot
s(J) < \frac{(Q\cdot \CP)^2+(Q\cdot 
  s(J))^2}{1-(\CP\cdot s(J))^2}<0. 
\ee
Therefore, if $\sgn(\CP\cdot Q)\,\sgn(\mathcal{I}(Q_1,Q_2;iJ))<0$
then $Q^2_J<0$. If $Q$ is a linear
combination of $\CP$ and $s(J)$, the determinant is zero. From this
follows that $Q^2_J=0$ only for $Q=0$, and otherwise $Q_J^2<0$. The
sum for $\sgn(Q\cdot \CP)\,\sgn(Q\cdot J)<0$ is therefore convergent.

What is left is the case $>0$. Then all the terms  vanish
identically, and therefore the whole sum is convergent. Inclusion of
$B$ and $C$ does not alter the final conclusion. 
\end{proof}

\begin{proposition}
\label{prop:1}
$\Phi_{\mu_{1\oplus 2}}^*(\tau,C,B)$ transforms under the generators $S$ and $T$
of $SL(2,\mathbb{Z})$ as:
\begin{eqnarray}
&S:&\quad \Phi_{\mu_{1\oplus 2}}^*(-1/\tau,-B,C)=-\frac{i(-i\tau)^{1/2}(i\bar
  \tau)^{b_2-1/2}}{\sqrt{|\Lambda_1^*/\Lambda_1||\Lambda_2^*/\Lambda_2|}}
  e(-P_{1\oplus 2}^2/4)\non \\
&& \qquad \qquad \qquad \qquad \qquad \sum_{\nu_{1\oplus 2} \in
    \Lambda^*_{1\oplus 2}/\Lambda_{1\oplus 2}}e(-\mu_{1\oplus 2}\cdot \nu_{1\oplus 2})\,\Phi_{\nu_{1\oplus 2}}^*(\tau,C,B),\non \\
&T:&\quad \Phi_{\mu_{1\oplus 2}}^*(\tau+1,B+C,B)=e((\mu_{1\oplus 2}+P_{1\oplus 2}/2)_{1\oplus 2}^2/2)\,\Phi_{\mu_{1\oplus 2}}^*(\tau,C,B),
  \non 
\end{eqnarray}
\end{proposition}

\begin{proof}
The $S$-transformation is proven using $\sum_{k\in
  \Lambda}f(k)=\sum_{k\in \Lambda^*}\hat f(k)$, with $\hat f(k)$ the
Fourier transform of $f(k)$. Therefore, one needs to determine the
following Fourier transform: 
\begin{eqnarray}
\label{eq:fouriertransform}
&&\int_{\Lambda_{1\oplus 2} \otimes \mathbb{R}}
d^{2b_2}\!x\,E\left(\mathcal{I}(x_1,x_2;iJ)
\sqrt{2\im(-1/\tau)}\right)\non\\
&&\quad \times\exp \left(\pi i\re(-1/\bar\tau)x_{1\oplus 2}^2+\pi \im(-1/\bar \tau)(\, x_{1\oplus 2}^2-2x_+^2\,)+2\pi i\,x\cdot y  \right) \\
&&=\int_{\Lambda_{1\oplus 2} \otimes \mathbb{R}}
d^{2b_2}\!x\,E\left(\mathcal{I}(x_1,x_2;iJ)
\sqrt{2\im(-1/\tau)}\right)\,e\left(-x_+^2/2\tau-(x^2_{1\oplus 2}-x_+^2)/2\bar\tau+x\cdot y\right),\non
\end{eqnarray}
and the one with $\mathcal{I}(x_1,x_2;iJ)$ replaced by
$Q\cdot\CP$. The following concentrates on the case with
$\mathcal{I}(x_1,x_2;iJ)$, the derivation for $Q\cdot\CP$ is
completely analogous. 

Let  $Q\cdot s(J)=\mathcal{I}(Q_1,Q_2;iJ)$ as in Proposition
\ref{prop:0}, then the following definite quadratic forms can
be defined: 
\be
Q_{1\oplus 2+}^2=Q_+^2+ (Q\cdot s(J))^2, \qquad
Q^2_{1\oplus 2-}=Q^2_{1\oplus 2}-Q_{1\oplus 2+}^2, \non
\ee
since $(J,J)\cdot s(J)=0$. Using these quadratic forms, we write
\be
e\left(-x_+^2/2\tau-(x^2_{1\oplus 2}-x_+^2)/2\bar\tau
\right)=e\left(-x_+^2/2\tau-(x\cdot s(J))^2/2\bar\tau-x^2_{1\oplus 2-}/2\bar\tau \right)\non
\ee
The Fourier transform can be written in the form
\begin{eqnarray}
&&=e\left(\tau y_+^2/2 +\bar \tau \mathcal{I}(y_1,y_2;iJ)^2/2 + \bar
  \tau y_{1\oplus 2-}^2\right) \non \\
&&\,\,\times \int_{\Lambda_{1\oplus 2} \otimes \mathbb{R}}
d^{2b_2}\!x\,E\left(\mathcal{I}(x_1,x_2;iJ)
\sqrt{2\im(-1/\tau)}\right) \non \\
&&\,\,\times e\left(-(x-y\tau)_+^2/2\tau-\mathcal{I}(x_1-y_1\bar \tau,x_2-y_2\bar \tau;iJ)^2/2\bar
\tau-(x-y\bar \tau)_{1\oplus
  2-}^2/2\bar \tau\right).\non
\end{eqnarray}
To proceed, one calculates the derivative of the integral
\begin{eqnarray}
&&\frac{\partial}{\partial \mathcal{I}(y_1,y_2;iJ)} \int_{\Lambda_{1\oplus 2} \otimes \mathbb{R}}d^{2b_2}\!x\,E\left(\mathcal{I}(x_1,x_2;iJ)
\sqrt{2\im(-1/\tau)}\right) \non \\
&&\qquad \qquad \times e\left(-(x-y\tau)_+^2/2\tau-\mathcal{I}(x_1-y_1\bar \tau,x_2-y_2\bar \tau;iJ)^2/2\bar
\tau-(x-y\bar \tau)_{1\oplus
  2-}^2/2\bar \tau \right) \non \\
&&= -\frac{i(-i\tau)^{1/2}(i\bar
  \tau)^{b_2-1/2}}{\sqrt{|\Lambda_1^*/\Lambda_1||\Lambda_2^*/\Lambda_2|}}\,\frac{\partial
  E\left(\mathcal{I}(y_1,y_2;iJ)
\sqrt{2\tau_2}\right)}{\partial \mathcal{I}(y_1,y_2;iJ)}.\non
\end{eqnarray}
This is shown by replacing the derivative by $-\bar
\tau \partial_{\mathcal{I}(x_1,x_2;iJ)}$, acting only on the
exponent; and performing a partial integration. The equality is then easily established. 
Since (\ref{eq:fouriertransform}) is an odd function of $y$, the
integration constant is 0. Therefore (\ref{eq:fouriertransform}) is
equal to
\begin{eqnarray}
\label{eq:fouriertransform5}
&&-\frac{i(-i\tau)^{1/2}(i\bar \tau)^{b_2-1/2}}{\sqrt{|\Lambda_1^*/\Lambda_1||\Lambda_2^*/\Lambda_2|}}E\left(\mathcal{I}(y_1,y_2;iJ)
\sqrt{2\tau_2}\right)  \\
&&\quad\times e\left(\tau y_+^2/2 +\bar \tau \mathcal{I}(Q_1,Q_2;iJ)^2/2 + \bar
  \tau y_{1\oplus 2-}^2\right) \non 
\end{eqnarray}
Using the standard techniques to include $B$- and $C$-field dependence
etc., one finds the posed transformation law. Note that
$\CP\cdot (Q_1-BP_1,Q_2-BP_2)=\CP\cdot (Q_1,Q_2)=\CP\cdot Q$. The proof of the
$T$-transformation is standard. 
\end{proof}

\begin{proposition}
\label{prop:shadow}
Define $\mathcal{D}=\partial_\tau
+\frac{i}{4\pi}\partial^2_{C_+}+\frac{1}{2}B_+\cdot\partial_{C_+}-\frac{1}{4}\pi i
B_+^2$, then
\be
\tau_2^{1/2}\,\mathcal{D}\Phi_{\mu_{1\oplus 2}}(\tau, C, B) \non
\ee
is a modular form of weight $(2,b_2-1)$.
\end{proposition}

\begin{proof}
The action of $\mathcal{D}$ on the exponents vanishes, and therefore
only the derivative to $\tau$ on the functions $E(z\sqrt{2\tau_2})$ remains. The proposition
follows easily from here.
\end{proof}

\begin{definition}
\label{def:2}
With the same input as for Definition \ref{def:1}:
\begin{eqnarray}
&&\Psi^*_{\mu_{1\oplus 2}}(\tau, C,B)=\non \\
&&\qquad \textstyle{\frac{1}{2\pi
  \sqrt{2\tau_2}}}\left(\textstyle{\sqrt{\frac{P\cdot
      J^2\,(P_1P_2J)^2}{P_1\cdot 
  J^2\,P_2\cdot
  J^2}}}\Theta_{\mu_{1}}(\tau,C,B)\,\Theta_{\mu_{2}}(\tau,C,B) - \textstyle{\sqrt{PP_1P_2}}\,\Theta_{\mu_{1\oplus
    2}}(\tau,C,B,\CP)\right)\non\\
&&\qquad +\textstyle{\frac{1}{2}}\sum_{Q\in \Lambda_{1
  \oplus 2}+\mu_{1\oplus 2}+P_{1\oplus 2}/2} \,(-1)^{P_{1}\cdot
  Q_1+P_{2}\cdot Q_2}\,(P_1\cdot Q_2-P_2\cdot Q_1) \\ 
&&\qquad \times\left(\, E\left(\mathcal{I}(Q_1,Q_2;t)
\sqrt{2\tau_2}\right)-E\left(\CP\cdot Q
\sqrt{2\tau_2}\right)\,\right)\non \\
&&\qquad \times \,e\left(  \tau (Q-B)_+^2/2+\bar \tau ((Q-B)^2_{1\oplus
  2}-(Q-B)_+^2)/2+(Q-B/2)\cdot C \right) \nonumber 
\end{eqnarray}
with $\Theta_{\mu_{i}}(\tau,C,B)$ as defined by
Eq. (\ref{eq:thetafunction}), summing over
$\Lambda_i$. $\Theta_{\mu_{1\oplus 2}}(\tau,C,B,\CP)$ is defined by 
\begin{eqnarray}
&&\Theta_{\mu_{1\oplus 2}}(\tau,C,B,\CP)=
\sum_{Q\in\Lambda_{1\oplus 2}+P_{1\oplus 2}/2+\mu_{1\oplus
    2}}(-1)^{P_{1\oplus 2}\cdot Q} \non \\ 
&&\qquad \times e\left(\tau (Q-B)_+^2/2+ \tau
(\CP\cdot Q)^2/2 + \bar \tau (Q-B)_{1\oplus 2 -}^2/2+ C\cdot (Q-B/2)
\right). \non
\end{eqnarray}
In the limit $\tau_2 \to \infty$, $\Psi^*_{\mu_{1\oplus
    2}}(\tau, C,B)$  approaches $\Psi_{\mu_{1\oplus
    2}}(\tau, C,B)$, which is defined in Eq. (\ref{eq:Zterm}). This
series is convergent because $\Phi^*_{\mu_{1\oplus 2}}(\tau, C,B)$
is convergent. 
\end{definition}

\begin{proposition}
$\Psi_{\mu_{1\oplus 2}}^*(\tau, C, B)$  transforms under the generators $S$ and $T$
of $SL(2,\mathbb{Z})$ as:
\begin{eqnarray}
&S:&\quad \Psi_{\mu_{1\oplus 2}}^*(-1/\tau,-B,C)=-\frac{(-i\tau)^{1/2}(i\bar
  \tau)^{b_2+1/2}}{\sqrt{|\Lambda_1^*/\Lambda_1||\Lambda_2^*/\Lambda_2|}}
  e(-P_{1\oplus 2}^2/4)\non \\
&& \qquad \qquad \qquad \qquad \qquad \sum_{\nu_{1\oplus 2} \in
    \Lambda_{1\oplus 2}^*/\Lambda_{1\oplus 2}}e(-\mu_{1\oplus 2}\cdot \nu_{1\oplus 2})\,\Psi_{\nu_{1\oplus 2}}^*(\tau,C,B),\non \\
&T:&\quad \Psi_{\mu_{1\oplus 2}}^*(\tau+1,B+C,B)=e((\mu_{1\oplus 2}+P_{1\oplus 2}/2)_{1\oplus 2}^2/2)\,\Psi_{\mu_{1\oplus 2}}^*(\tau,C,B),\non
\end{eqnarray}
\end{proposition}

\begin{proof}
This is a continuation of the proof of Proposition \ref{prop:1}. The following
Fourier transform needs to be calculated: 
\begin{eqnarray}
\label{eq:fouriertransform2}
&&\int_{\Lambda_{1\oplus 2} \otimes \mathbb{R}}
d^{2b_2}\!x\,(P_1\cdot x_{2}-P_2\cdot x_{1})\,E\left(\mathcal{I}(x_1,x_2;iJ)
\sqrt{2\im(-1/\tau)}\right)  \\
&&\qquad \times e\left(-(x^2_{1\oplus 2}-x_+^2)/2\bar\tau-x_+^2/2\tau+x\cdot y\right),\non
\end{eqnarray}
and the one with $\mathcal{I}(x_1,x_2;iJ)$ replaced by $\CP\cdot
Q$. We again concentrate on the case with $\mathcal{I}(x_1,x_2;iJ)$. It is
instructive to write $P_1\cdot x_{2}-P_2\cdot x_{1}$ as
$(-P_2,P_1)\cdot x^\mathrm{T}$ with $x=(x_1,x_2)$. The inner product
$(-P_2,P_1)\cdot x_+$ with $x_+=x\cdot J J/P\cdot J^2$ vanishes. Therefore, 
\begin{eqnarray}
(-P_2,P_1)\cdot x^\mathrm{T}&=&(-P_2,P_1)\cdot
x^\mathrm{T}_-+(-P_2,P_1)\cdot s(J)^\mathrm{T} \,x\cdot s(J)  \non \\
&=&(-P_2,P_1)\cdot x^\mathrm{T}_-+\textstyle{\sqrt{\frac{P\cdot J^2\,(P_1P_2J)^2}{P_1\cdot
  J^2\,P_2\cdot J^2}}}\,\mathcal{I}(x_1,x_2;iJ), \non
\end{eqnarray}
with $s(J)\in \Lambda_{1\oplus 2}$ as in the proof of Proposition
\ref{prop:0}. This shows that the factor $P_1\cdot x_{2}-P_2\cdot x_{1}$ can be
replaced by $(2\pi i)^{-1}\left(\,(-P_2,P_1)\cdot \partial_{y_-}+\textstyle{\sqrt{\frac{P\cdot J^2\,(P_1P_2J)^2}{P_1\cdot
  J^2\,P_2\cdot J^2}}} \partial_{\mathcal{I}(y_1,y_2;iJ)}\,\right)$. Using
Proposition \ref{prop:1}, one finds that (\ref{eq:fouriertransform2}) equals 
\begin{eqnarray}
&-\frac{(-i\tau)^{1/2}(i\bar
    \tau)^{b_2+1/2}}{\sqrt{|\Lambda_1^*/\Lambda_1||\Lambda_2^*/\Lambda_2|}}
&\left[ (P_1\cdot y_{2}-P_2\cdot y_{1})\,E\left(\mathcal{I}(y_1,y_2;iJ)
\sqrt{2\tau_2}\right) e\left(\tau y^2_+/2 +\bar \tau (y_{1\oplus
  2}^2-y_+^2)/2 \right) \right. \non \\
&\qquad&\,\,+\frac{\sqrt{2\tau_2}}{\pi i \bar \tau}\textstyle{\sqrt{\frac{P\cdot J^2\,(P_1P_2J)^2}{P_1\cdot
  J^2\,P_2\cdot J^2}}}\, \left.e\left(\tau y^2_+/2 +\tau
\mathcal{I}(y_1,y_2;iJ)^2/2 +\bar \tau y_{1\oplus 2-}^2/2 \right)
\,\right]. \non
\end{eqnarray}
Clearly, this Fourier transform leads to a shift in the modular
transformation properties. This can be cured if one recalls
the transformation properties of the second Eisenstein series:
$E_2(-1/\tau)=\tau^2\,(E_2(\tau)-\frac{6i}{\pi\tau})$. A correction
term can be added to $E_2(\tau)$: $E^*_2(\tau)=E_2(\tau)-\frac{3}{\pi
  \tau_2}$ which transforms as a modular form of weight 2. This leads
precisely to the term with theta functions in the definition. This
means that the discontinuous function $z\,\sgn(z)$, which appears in
(\ref{eq:Zterm}), is replaced in $\Psi^*_{\mu_{1\oplus2}}(\tau,C,B)$ by
the real analytic function $F(z)=z\,E(z)+\frac{1}{\pi}e^{-\pi z^2}$. $F(z)$ approaches
$z\,\sgn(z)$ for $|z|\to \infty$.
\end{proof}
\begin{proposition}
\label{prop:5}
With $\mathcal{D}$ as in Proposition \ref{prop:shadow}
\begin{eqnarray}
\mathcal{D}\Psi_{\mu_{1\oplus 2}}^*(\tau, C, B)&=&  -\frac{i}{2\sqrt{2\tau_2}}\textstyle{\sqrt{\frac{P\cdot J^2\,(P_1P_2J)^2}{P_1\cdot
  J^2\,P_2\cdot J^2}}}\Upsilon_{\mu_{1\oplus 2}}(\tau,C,B) \non\\
&& +\frac{i}{4\pi  (2\tau_2)^{3/2}}\left(\Theta_{\mu_{1}}(\tau,C,B)\Theta_{\mu_{2}}(\tau,C,B)-\Theta_{\mu_{1\oplus 2}}(\tau,C,B,\CP) \right),\non
\end{eqnarray}
with 
\begin{eqnarray}
&&\Upsilon_{\mu_{1\oplus 2}}(\tau,C,B)=\sum_{Q\in\Lambda_{1\oplus 2}+P_{1\oplus 2}/2+\mu_{1\oplus
    2}}(-1)^{P_{1\oplus 2}\cdot Q}\,(-P_2, P_1)\cdot Q_-\,\mathcal{I}(Q_1,Q_2;t)  \non
\\ 
&&\qquad \times e\left(\tau (Q-B)_+^2/2+ \tau
\mathcal{I}(Q_1,Q_2;t)^2/2 + \bar \tau (Q-B)_{1\oplus 2 -}^2/2+C\cdot
(Q-B/2) \right) \non
\end{eqnarray}
\begin{proof}
The proof is straightforward. Note that $\Theta_{\mu_{i}}(\tau,C,B)$ and $\Upsilon_{\mu_{1\oplus
    2}}(\tau,C,B)$ are not mock modular forms. The weights are
respectively $(1,b_2-1)$ and $(2,b_2)$, such that the weight of $\mathcal{D}\Psi_{\mu_{1\oplus 2}}^*(\tau, C, B)$ is
$(5/2, (2b_2+1)/2)$ as expected.
\end{proof}
\end{proposition}

\providecommand{\href}[2]{#2}\begingroup\raggedright

\end{document}